\documentclass[10pt]{amsart}
\usepackage{amsmath}
  \usepackage{graphics} 
\usepackage{graphicx}  
  \usepackage{epsfig} 
 \usepackage[colorlinks=true]{hyperref}
  \usepackage{float}
  \usepackage{caption}
\usepackage{amsmath}
\usepackage{wrapfig}
\usepackage{graphicx}
\usepackage{subcaption}
\usepackage{lscape}
\usepackage{amssymb}
\DeclareMathOperator{\Tr}{tr}

\usepackage{amsthm}
\theoremstyle{plain}

\newtheorem*{theorem*}{Theorem}
\newtheorem{theorem}{Theorem}[section]
\newtheorem{corollary}{Corollary}

\newtheorem{lemma}[theorem]{Lemma}
\newtheorem{proposition}{Proposition}
\newtheorem{conjecture}{Conjecture}
\newtheorem{assumption}{Assumption}

\theoremstyle{definition}
\newtheorem{definition}[theorem]{Definition}
\newtheorem{remark}{Remark}

\title[ additional food ] 
    {Additional food causes predator ``explosion" - unless the predators compete}

\begin{document}
\maketitle

\centerline{\scshape    Rana D. Parshad$^{1}$, Sureni Wickramsooriya$^{2}$, Kwadwo Antwi-Fordjour$^{3}$ and Aniket Banerjee$^{1}$  }
\medskip
{\footnotesize

   \medskip
   
    \centerline{ 1) Department of Mathematics,}
 \centerline{Iowa State University,}
   \centerline{Ames, IA 50011, USA.}
      \medskip
 \centerline{2) Department of Mathematics and Physics,}
 \centerline{Kansas Wesleyan University,}
   \centerline{ Salina, KS 67401, USA.}
   \medskip
 \centerline{3) Department of Mathematics and Computer Science,}
 \centerline{Samford University, }
   \centerline{ Birmingham, AL 35229, USA.}  
   \medskip
   
 }

\begin{abstract}

The literature posits that an introduced predator population, is able to drive it's target pest population extinct, if supplemented with high quality additional food of quantity $\xi > \xi_{critical}$, \cite{ SP11, SPV18, SPD17, SPM13}. We show this approach actually leads to infinite time \emph{blow-up} of the predator population, so is unpragmatic as a pest management strategy. 
We propose an alternate model in which the additional food induces predator competition. Analysis herein indicates that there are threshold values $c^{*}_{1} < c^{*}_{2} < c^{*}_{3}$ of the competition parameter $c$, s.t. when $c < c^{*}_{1}$, the pest free state is globally stable, when $c^{*}_{2} < c < c^{*}_{3}$, bi-stability is possible, and when $c^{*}_{3} < c$, up to three interior equilibria could exist. 
As $c$ and $\xi$-$c$ are varied standard, co-dimension one and co-dimension two bifurcations are observed. The recent dynamical systems literature involving predator competition, report several 
non-standard bifurcations such as the saddle-node-transcritical bifurcation (SNTC) occurring in co-dimension two \cite{KSV10, BS07}, and cusp-transcritical bifurcation (CPTC) in co-dimension three, \cite{D20, BS07}. We show that in our model structural symmetries can be exploited to construct a SNTC in co-dimension two, and a CPTC also in co-dimension \emph{two}. We further use these symmetries to construct a novel pitchfork-transcritical bifurcation (PTC) in co-dimension two, thus explicitly characterizing a new organizing center of the model. Dynamics such as homoclinic orbits, concurrently occurring limit cycles, and competition driven Turing patterns are also observed. Our findings indicate that increasing additional food in predator-pest models, can \emph{hinder} bio-control, contrary to some of the literature. However, additional food that also induces predator competition, leads to novel bio-control scenarios, and complements the work in \cite{H21, B98, K04, D20, BS07, VH19}.
\end{abstract}

 \section{Introduction}

The number of invasive species continues to grow worldwide at an unprecedented rate \cite{PZM05, PSCEWT16, L16, E11, S17}, thereby making the development and implementation of control strategies crucial \cite{S17}. Chemical pesticides, despite negative environmental and human health impacts are used heavily for such control, \cite{PB14}. This includes large scale use in the North Central United States, where losses due to invasive pests such as the European corn borer, Western corn root worm and Soybean Aphid exceed $\$$3 billion annually \cite{L15, TS14, L14, CO18}.
An alternative to pesticides is classical bio-control - introduction of natural enemies of the pest species \cite{V96, C15, B07, K17, S99}. The benefits of this approach is that it is non-toxic and can be self-sustainable, removing the need for the repeated applications often required with chemical pesticides \cite{B07, V96}. 

One approach when the introduced predator does not sufficiently reduce pest density, is to boost predator efficacy by supplementing it with an additional food source \cite{CO98, ES93, SC07, S10, W16, T15, R95, T15, ES93, SV06}. Yet, there is no clear consensus in the biological literature, as to why or why \emph{not} this approach works. 
A number of mathematical models that describe predator-pest dynamics with an additional food source have been developed, \cite{SP07, SP10, SP11, SPV18, SPD17, SPM13, VA22}. 
The works unanimously posit that if high quality additional food is provided to the introduced predator in quantity $\xi$, s.t. $\xi > \xi_{critical}$, (where $\xi_{critical}$ depends on other model parameters) then the outcome is pest extinction - In summary, this approach is claimed to \emph{always} enhance bio-control, see Table \ref{Table:3}. 

Bio-control is not without risk in practice, and can pose various challenges \cite{SV06, F14}. For example uncontrolled growth of the introduced predator itself - which could then lead to a further sequence of non-target effects, \cite{PQB16, PB16}. In such situations a natural self regulating mechanism for the ``growing" population, is intraspecies competition. There has been a large research focus on invasion control and management using niche theory, and competition theory in novel ways to manipulate  competition in invader-invadee systems \cite{britton2018trophic, huston2004management, kettenring2011lessons, jeremy10}. For example, at low nitrogen levels many native plants are better competitors than invasives, who are better competitors at higher levels \cite{kuang2010interacting}. Lowering nitrogen levels in soil via microbe use, has been used as a successional management strategy, to enhance intraspecies competition among native plants, and preempt invasions \cite{vasquez2008creating}. However, the flip side has been less investigated - that is strategies that aim to increase intraspecific competition among introduced predator populations, as a means of enhancing bio-control. So the responses of invasive pests to increased   competition, among introduced predators and their subsequent effects are less understood  \cite{mwangi2007niche, hille12, seabloom2003invasion}. Predator competition, in general predator-prey models, can greatly enhance the variety of the possible dynamics. This was first investigated qualitatively \cite{K04, B98}, and several higher order bifurcations were uncovered. However, only very recently have these been rigorosly proved, \cite{H21}. For applications of similar methods used in \cite{H21}, to type III and IV responses, see \cite{JH13, JH14, JH04}. Note that no current models of additional food mediated bio-control, assume competition among the introduced predators. The current manuscript is a modeling endeavor to this end.
Our findings in the current manuscript are as follows,
\begin{itemize}

\item Increasing the quantity of additional food $\xi$, s.t $\xi > \xi_{critical}$, in classical predator-pest models, causes \emph{unbounded} growth of the predator population density, for any initial condition, shown via Theorem \ref{thm:t1ug}. 

\item Thus the strategy of supplementing introduced predators with sufficient quantities of additional food (which the literature is rife with, see Table \ref{Table:3}), 
can actually \emph{hinder} bio-control - as predators in excessively high density can cause a host of non-target effects.

\item A new model for additional food mediated bio-control is proposed, via \eqref{Eqn:1nn}.
Herein it is assumed the additional food also induces intraspecific predator competition.

\item The analysis of \eqref{Eqn:1nn} reveals novel bio-control possibilities, driven by competition. Now a bounded pest extinction state is possible. This state could be globally stable or bi-stable, shown via Theorems \ref{thm:1o3}, \ref{thm:gloint}, \ref{thm:glopestfree}.

\item  A rich dynamical structure for \eqref{Eqn:1nn} is revealed. We report, 
\newline
(i) Hopf, saddle-node and transcritical bifurcations in co-dimension one, see Fig. \ref{fig:Hopf-SN-RegimesA}.
\newline
(ii) A saddle-node-transcritical bifurcation in co-dimension two, see Fig. \ref{fig:SNTC}.
\newline
(iii) A cusp-transcritical bifurcation, in co-dimension \emph{two}, see Fig. \ref{fig:Cusp-Transcritical}.
\newline
(iv) A pitchfork-transcritical bifurcation, in co-dimension two, see Fig. \ref{fig:Pitchfork-Transcritical}.
\newline
(v) Two concurrently occuring limit cycles, and homoclinic orbits, see Fig. \ref{fig:limit cycles and homoclinic}.
\newline
(vi) Competition driven Turing instability, see Fig. \ref{fig:Turing pattern}.

\item The consequences of these new dynamics to bio-control are discussed in detail.
Several new conjectures are made about the dynamical consequences of predator competition, 
see Conjecture \ref{conj:c11}.

\end{itemize}

\section{Prior Results}

\subsection{General Predator-Pest Model}
The following \emph{general} model for an introduced predator population $y(t)$ depredating on a target pest population $x(t)$, while also provided with an additional food source of quality $\frac{1}{\alpha}$  and quantity $\xi$, has been proposed in the literature, 
\begin{equation}
\label{Eqn:1g}
\frac{dx}{dt} = x\left(1-\frac{x}{\gamma}\right) - f(x,\xi,\alpha) y, \    \frac{dy}{dt} =  g(x, \xi, \alpha) y  - \delta y.
\end{equation}
Here $f(x,\xi,\alpha)$ is the functional response of the predator, that is pest dependent but also dependent on the additional food (hence the explicit dependence on $\xi, \alpha$). Likewise,
$g(x,\xi,\alpha)$ is the numerical response of the predator. 
If $\xi = 0$, that is there is no additional food, the model reduces to a classical predator-prey model of Gause type, that is $f(x,0,0)= g(x,0,0)$, where $f$ has the standard properties of a functional response. For these models we know pest eradication \emph{is not possible}, as the only pest free state is $(0,0)$, which is typically unstable \cite{K01}. Thus modeling the dynamics of an introduced predator and its prey, a targeted pest via this approach, where the constructed 
$f(x,\xi,\alpha), g(x,\xi,\alpha)$ are used as a means to achieve a pest free state, has both practical and theoretical value, and thus has been well studied. Table \ref{Table:3},  summarizes some of the key literature, in terms of the functional forms used in these models.

\begin{table}[H]
\caption{Dynamics of predator-pest models supplemented with AF}\label{Table:3}
	\scalebox{0.72}{
{\begin{tabular}{|c|c|c|c|c|}
		\hline
		& \mbox{Functional form} & \mbox{Relevant literature} & \mbox{AF requirement}& \mbox{ Effect on pest control}\\
		& &  & &  \\ \hline \hline
		(i) & $f(x,\xi,\alpha) = \frac{x}{1+\alpha \xi + x}$ & \cite{SP07, SP10, SP11} & $\xi > \frac{\delta}{\beta - \delta \alpha}$ & \mbox{Pest is eradicated, switching AF}\\
		& $g(x,\xi,\alpha) = \frac{\beta (x+\xi)}{1+\alpha \xi + x}$ &  &  &  \mbox{maintains/eliminates predator }\\          \hline
		(ii) & $f(x,\xi,\alpha) = \frac{x^{2}}{1+\alpha \xi^{2} + x^{2}}$  & \cite{SPV18}  & $\xi >\sqrt{ \frac{\delta}{\beta - \delta \alpha}}$  &  \mbox{Pest is eradicated, switching AF}\\
		& $g(x,\xi,\alpha) = \frac{\beta (x^{2}+\xi^{2})}{1+\alpha \xi^{2} + x^{2}}$ &  & & \mbox{ maintains/eliminates predator } \\ \hline
		(iii) & $f(x,\xi,\alpha) = \frac{x}{(1+\alpha \xi)( \omega x^{2}+1) + x}$   & \cite{SPD17} & $\xi > \frac{\delta}{\beta - \delta \alpha}$ &  \mbox{Pest is eradicated, switching AF} \\
		& $g(x,\xi,\alpha) = \frac{\beta (x+\xi(\omega x^{2}+1)}{(1+\alpha \xi)(\omega x^{2} + 1) + x}$ &  &  & \mbox{ maintains/eliminates predator}  \\ \hline
		(iv) & $f(x,\xi,\alpha) = \frac{x}{1+\alpha \xi + x + \epsilon y}$ & \cite{SPM13} &  $\beta \xi = \delta(1+\alpha)$  & \mbox{Pest extinction state stabilises,}  \\
		& $g(x,\xi,\alpha) = \frac{\beta (x+\xi)}{1+\alpha \xi + x + \epsilon y}$ &  & $\beta(\gamma + \xi) =\delta(1+\alpha \xi + \gamma)$  &\mbox{via TC}  \\ \hline
		\end{tabular}}
		}
\vspace{.1cm}		 
     \flushleft
\textbf{Note:}	Herein $\gamma$ is the carrying capacity of the pest, $\beta$ is the conversion efficiency of the predator, $\delta$ is the death rate of the predator, $\frac{1}{\alpha}$ is the quality of the additional food provided to the predator and $\xi$ is the quantity of additional food provided to the predator. 
		\end{table}

\subsection{Key idea behind extinction mechanism}
		The key idea in all of the works in Table \ref{Table:3}, is to implement a pest management strategy, where  the quantity of additional food $\xi$ is increased, so that $\xi > \xi_{critical}$. Here $\xi_{critical}$ depends on the model parameters, and changes from model to model. The dynamical effect of this action is to push the vertical predator nullcline past the y-axis (predator axis), into the $2^{nd}$ quadrant, see Fig. \ref{fig:Blow-up}. Thus, now there is no positive interior equilibrium. Via positivity of solutions, trajectories will move towards the predator axis and ``hit" it, yielding pest extinction. 
\begin{figure}[!htb]
\begin{center}
    \includegraphics[width=5cm, height=5.4cm]{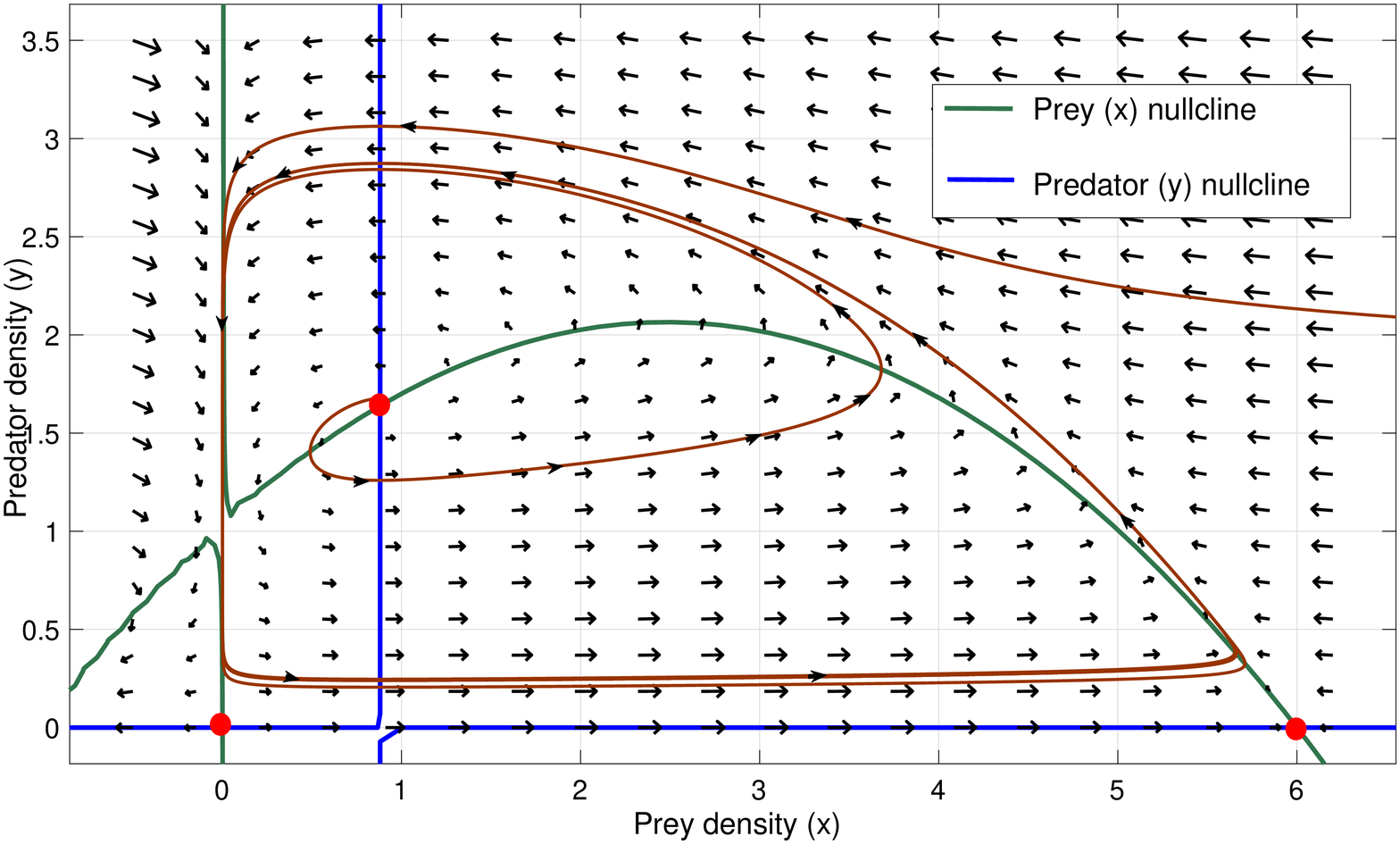} 
    \includegraphics[width=5cm, height=5.4cm]{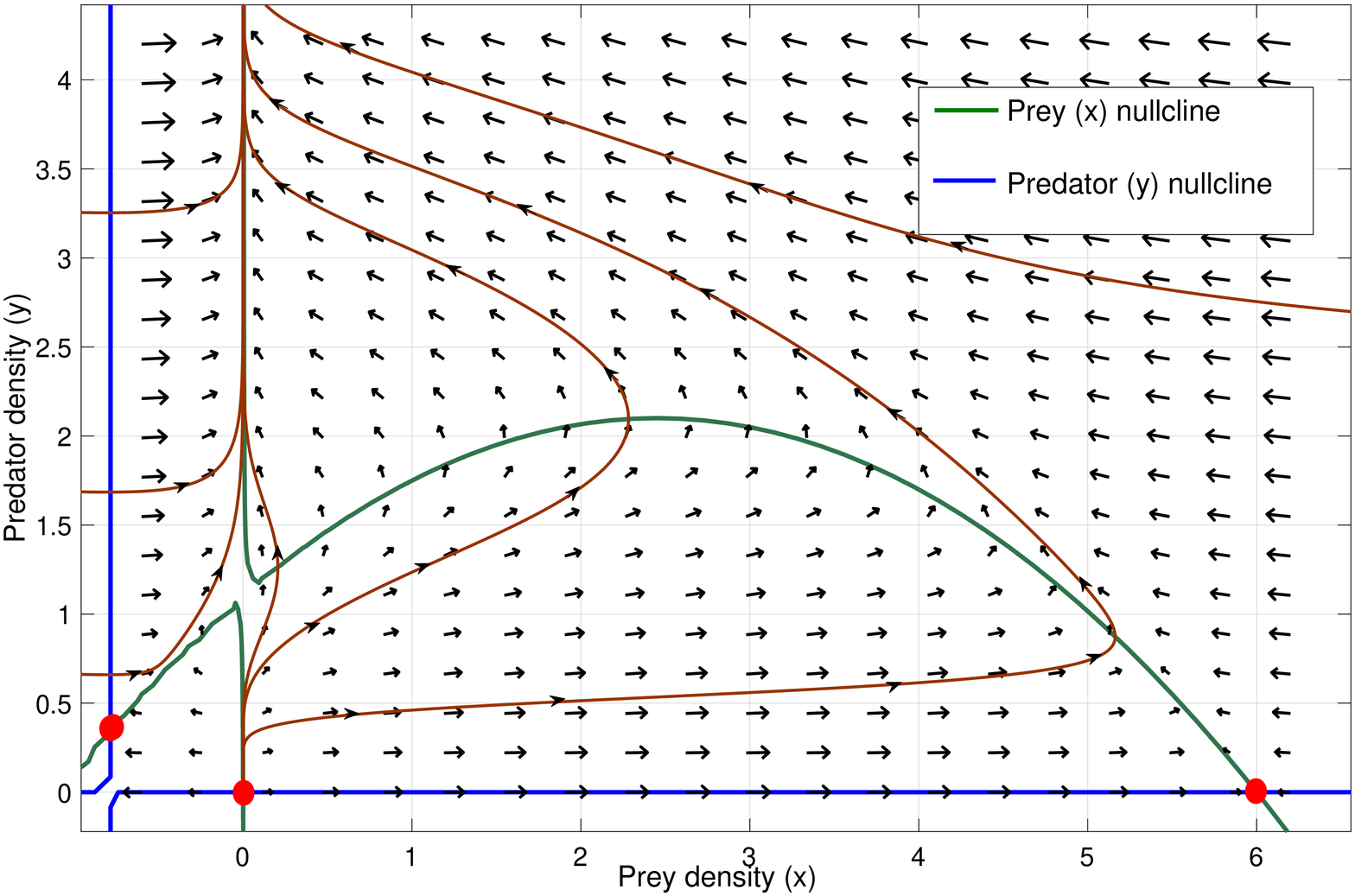}  
\end{center}
\caption{We compare $\xi <(>) \xi_{critical}$. (a) (or left figure)  $\xi = 0.5 < 0.7 = \frac{\delta}{\beta - \delta \alpha}$, and we see a large amplitude limit cycle . Here $\gamma=6, \alpha=0.1, \beta=0.3, \delta=0.2$ and $\xi=0.5$ (b) (or right figure) Now $\xi$ is increased, so $\xi = 1 > 0.7 = \frac{\delta}{\beta - \delta \alpha}$ - notice the predator nullcline is at $x=-1$. We see trajectories tend upwards along the y-axis (predator axis) - what we will show is that they tend to $(0,\infty)$, see Theorem \ref{thm:t1ug}.}
\label{fig:Blow-up}
\end{figure}

The literature also claims this is beneficial from a management standpoint, as one can control the additional food quantity $\xi$ to enhance pest control. That is,
\begin{itemize}
\item[(i)] $\xi$ is stopped to cause eventual predator extinction, after pest extinction, \cite{SP07}.
\item[(ii)] $\xi$ is kept at a requisite level, to maintain the predator at a target level, \cite{SP11, SP18}.
\item[(iii)] $\xi$ is changed so target states can be reached in minimal time, \cite{SP10, SP11, VA22}.
\end{itemize}

\vspace{.5cm}
\noindent We focus on the case with type II response next.

\subsection{Type II response}
The basic predator-pest model, with type II response is as follows,

\begin{equation}
\label{Eqn:1}
\frac{dx}{dt} = x\left(1-\frac{x}{\gamma}\right) - \frac{xy}{1+\alpha \xi + x}, \    \frac{dy}{dt} =\frac{\beta xy}{1+\alpha \xi + x} + \frac{\beta \xi y}{1+\alpha \xi + x} - \delta y.
\end{equation}

We recap the key results from the literature. These are essentially categorized based on quality and quantity of additional food. We will focus on the case of ``high quality" additional food, 
that is if $\frac{1}{\alpha} > \frac{\delta}{\beta} \Leftrightarrow \beta - \delta \alpha > 0$. Two regimes of the quantity of additional food are possible, detailed next.

\subsubsection{The small additional food regime}
We first consider $\xi \in (0, \frac{\delta}{\beta - \delta \alpha})$, the ``small" additional food regime. We further assume $\beta - \delta > 0$, this is biologically realistic, and follows from standard theory when $\xi=0$, and \eqref{Eqn:1} reduces to the standard Rosenzweig McArthur model \cite{K01}. In this case, there is one interior equilibrium which is stable if, $\frac{\delta}{\beta - \delta} < \gamma \leq \frac{\beta + \delta}{\beta - \delta}$, and predator and pest co-exist.
Or, the equilibrium could be unstable, with the existence of a stable limit cycle, if $\frac{\delta}{\beta - \delta}  \leq \frac{\beta + \delta}{\beta - \delta} < \gamma$.
That is a Hopf Bifurcation occurs at $\frac{\beta + \delta}{\beta - \delta} = \gamma$. In this case the predator and pest would now exhibit cyclical dynamics \cite{SP07}. Increasing additional food in the interval $ (0, \frac{\delta}{\beta- \delta \alpha}]$, lowers pest equilibrium (as essentially this moves the vertical predator nullcline closer and closer to the y-axis).
However, the pest cannot be driven extinct for any $\xi \in (0, \frac{\delta}{\beta - \delta \alpha})$.

%

\subsubsection{The large additional food regime}

The literature \cite{SP07, SP11, SP10} claimed that increasing the additional food supply beyond this interval, so if $\xi \in \left(\frac{\delta}{\beta- \delta \alpha}, \infty \right)$,
eradicates the pest from the ecosystem in a finite time, and from that time the predators survive only on the additional food supply. We recap the result of interest from the literature \cite{SP07, SP11}, which quantifies the efficacy of the predator to achieve pest eradication when supplemented with additional food via \eqref{Eqn:1},

\begin{lemma}
\label{lem:1}

Consider the predator-pest system described via \eqref{Eqn:1}.
 If the quality of the additional food satisfies $\beta - \delta \alpha  > 0$, then the pest can be eradicated from the ecosystem in a finite time by providing the predator with additional
food of quantity $\xi > \frac{\delta}{\beta - \delta \alpha }$.


\end{lemma}

However, this is not an accurate representation of the dynamics of \eqref{Eqn:1}. 
We recap the following [Theorem 2.2, \cite{PWB20}],

\begin{theorem*}
Consider the predator-pest system described via \eqref{Eqn:1}. Pest eradication is not possible in finite time even if the quality of the additional food satisfies $\beta - \delta \alpha  > 0$
and the quantity of the additional food satisfies $\xi > \frac{\delta}{\beta - \delta \alpha}$.
\end{theorem*}

\begin{remark}
Although pest extinction is not possible in finite time, [Theorem 2.2, \cite{PWB20}], we have shown pest extinction is possible in infinite time, [Proposition 1, \cite{PWB20}] - here we show the pest can decay to the extinction state, exponentially $(\approx e^{-t})$. 
\end{remark}

Note, the resulting dynamics of the predator population $y$, in the case of pest extinction is \emph{not} explored in \cite{PWB20}. Neither has this dynamic been rigorously explored in \emph{any} of the literature on additional food models \cite{SP07, SP10, SP11, SPV18, SPM13, SPD17}.
What we find, is that in the case of pest extinction, the resulting dynamics of the predator could \emph{hinder} bio-control - and so is altogether unpragmatic as a pest management strategy. We provide details to this end next.


\subsection{Unbounded Predator growth}
We see that providing additional food in the large regime, can drive the pest extinct asymptotically (c.f. [Proposition 1, \cite{PWB20}]) - but only at the cost of \emph{unbounded} growth of the introduced predator. That is, if $\xi > \frac{\delta}{\beta - \delta \alpha}$, the dynamics of \eqref{Eqn:1} always result in blow-up in infinite time of the predator density. To this end we state the following theorem,
 

\begin{theorem}
\label{thm:t1ug}
Consider the predator-pest system described via \eqref{Eqn:1}, and assume the parametric restrictions, $\beta - \delta \alpha >0$. If the quantity of additional food $\xi$ is chosen s.t., 
$\xi \in \left( \frac{\delta}{\beta - \delta \alpha}, \infty \right)$, then solutions initiating from any positive initial condition, blow up in infinite time. That is,

\begin{equation*}
(x(t),y(t)) \rightarrow (0,\infty) \ \mbox{as} \ t \rightarrow \infty.
\end{equation*}

\end{theorem}

\begin{proof}

Note, WLOG $\beta > \delta$, and so if $\xi > \frac{\delta}{\beta - \delta \alpha}$, then $\beta \xi  > \delta(1+\alpha \xi )$ 

or $(\beta x + \beta \xi)  > \delta(1+\alpha \xi + x)$. Thus,


\begin{equation*}
 \frac{dy}{dt} = \frac{\beta xy}{1+\alpha \xi + x} + \frac{\beta \xi y}{1+\alpha \xi + x} - \delta y = \left(\frac{(\beta x + \beta \xi) - \delta(1+\alpha \xi + x)}{1+\alpha \xi + x} \right) y \geq 0,
\end{equation*}

if $x,y>0$. First we consider the region

\begin{equation*}
\left(1-\frac{x}{\gamma}\right)(1+\alpha \xi + x) < y.
\end{equation*}

Herein
\begin{equation*}
\frac{dx}{dt} = x\left(1-\frac{x}{\gamma}\right) - \frac{xy}{1+\alpha \xi + x} < 0.
\end{equation*}

Thus in this region the $x$ component of trajectories is decreasing, as $\frac{dx}{dt}<0$ whereas the $y$ component is increasing as $\frac{dy}{dt} \geq 0$. Thus all trajectories move towards the y-axis as there is no  positive interior equilibrium, and the boundary equilibria $(0,0)$ is a source, while $(\gamma, 0)$ is a saddle, so both of those are non attracting. Next consider the region of the phase space,

\begin{equation*}
\left(1-\frac{x}{\gamma}\right)(1+\alpha \xi + x) > y.
\end{equation*}

Then $\frac{dx}{dt}>0$ and $\frac{dy}{dt} \geq 0$, thus trajectories will move outwards, till they enter the region $(1-\frac{x}{\gamma})(1+\alpha \xi + x) < y$. However, once they enter this region, $\frac{dx}{dt}<0$ and $\frac{dy}{dt} \geq 0$, and so trajectories will turn and move towards the y-axis.

Now given the parametric restriction $\xi > \frac{\delta}{\beta - \delta \alpha} \Leftrightarrow \left( \frac{\beta \xi }{1+\alpha \xi } - \delta \right) > 0$, $\exists~ \delta_{1} > 0$ s.t $\left( \frac{\beta \xi }{1+\alpha \xi } - \delta \right) >  \delta_{1}  > 0$. Thus for a given set of parameters s.t $\left( \frac{\beta \xi }{1+\alpha \xi } - \delta \right) >  \delta_{1}  > 0$, $\exists~ \epsilon(\delta_{1})$ s.t.

\begin{equation}
\label{eq:1e}
 \left( \frac{\beta \xi }{1+\alpha \xi + \epsilon} - \delta \right)> \frac{\delta_{1}}{2}.
\end{equation}

Once trajectories are close enough to the y-axis, that is when for an $\epsilon$ in \eqref{eq:1e}, $x < \epsilon$, from \eqref{Eqn:1} we have,

\begin{equation*}
 \frac{dy}{dt} =  \frac{\beta xy}{1+\alpha \xi + x} + \frac{\beta \xi y}{1+\alpha \xi + x} - \delta y \geq \left( \frac{\beta \xi }{1+\alpha \xi + \epsilon} - \delta \right) y > \left( \frac{\delta_{1}}{2}\right) y.
\end{equation*}

 Via simple comparison $y > e^{\left(\frac{\delta_{1}}{2}\right)  t}y_{0}$, so trajectories will move upwards along the y-axis. 
 
Note, a $(0,y^{*})$ state is clearly not achievable. Assume it was, we could consider an initial condition $(x_{0}, y_{0})$ with $x_{0} < \gamma$, $y_{0} >> \max \left(y^{*}, \frac{(1+\alpha \xi + \gamma)^{2}}{4\gamma}\right)$. Here $\frac{(1+\alpha \xi + \gamma)^{2}}{4\gamma}$, is the maximum value of the prey nullcline. Since $\frac{dx}{dt} > 0, \frac{dy}{dt} \geq 0$, for solutions initiating from this initial condition, trajectories cannot turn downwards towards the alleged steady state $(0,y^{*})$. This yields a contradiction. 

Thus the only eventuality is that,

\begin{equation*}
(x(t),y(t)) \rightarrow (0,\infty) \ \mbox{as} \ t \rightarrow \infty,
\end{equation*}

that is the predator $y$ blows up in infinite time. This proves the Theorem.

\end{proof}


\begin{remark}
Note, a stable (and bounded) pest free state is only possible if $\xi = \left( \frac{\delta}{\beta- \delta \alpha} \right)$.
\end{remark}

\begin{figure}[!htb]
\begin{center}
      \includegraphics[width=6cm, height=6cm]{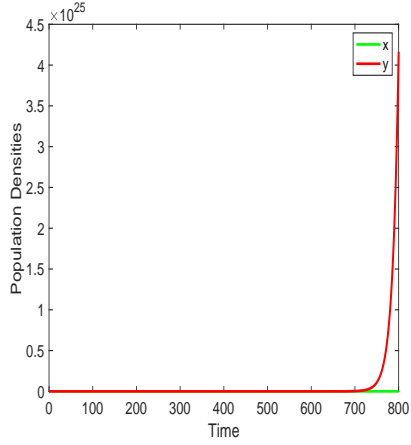}   
\end{center}
\caption{Time series of the population densities depicting the blow up of the predator in infinite time with IC$=[3,2]$. Here $\gamma=6, \alpha=0.1, \beta=0.3, \delta=0.2$ and $\xi=1$. }
\label{fig:Blow-up}
\end{figure}

\section{A New Modeling Formulation}

\subsection{Additional Food Mediated Competition}

Since unbounded growth of the predator established via theorem \ref{thm:t1ug} can lead to various non-target effects \cite{PQB16}, increasing additional food in the regime $\xi \in \left(\frac{\delta}{\beta- \delta \alpha}, \infty\right)$, via \eqref{Eqn:1} is not pragmatic as a pest control strategy. To circumvent this we introduce a new model. To this end we first make the following assumption,

\begin{assumption}
The additional food initiates intraspecific competition amongst the introduced predators, which is not present without the additional food.
\end{assumption}
 
 This leads to the following model,

\begin{equation}
\label{Eqn:1nn}
\frac{dx}{dt} = x\left(1-\frac{x}{\gamma}\right) - \frac{xy}{1+\alpha \xi + x}, \    \frac{dy}{dt} =\frac{\beta xy}{1+\alpha \xi + x} + \frac{\beta \xi y}{1+\alpha \xi + x} - \delta y - c\xi y^{2},
\end{equation}

with a classical quadratic competition term - which is dependent on the quantity additional food $\xi$. That is if there is no additional food ($\xi = 0$), then there is no intraspecific competition among the predators, and we recover the classical Rosenzweig-McArthur model.


The competition foremost prevents unbounded growth of the predator.

\begin{lemma}
\label{lem:ugp}
Consider the predator-pest system described via \eqref{Eqn:1nn}, and assume the parametric restrictions, $\beta - \delta \alpha >0$,
$\xi > \frac{\delta}{\beta - \delta \alpha}$, then all solutions $(x(t),y(t))$ initiating from any positive initial condition, remain bounded for all time.

\end{lemma}

\begin{proof}
We have from the equation for $y$, positivity of solutions, and the boundedness of the pest $x$ by its carrying capacity $\gamma$ that,
\begin{equation*}
 \frac{dy}{dt} = \frac{\beta xy}{1+\alpha \xi + x} + \frac{\beta \xi y}{1+\alpha \xi + x} - \delta y -c\xi y^{2} < \left(\frac{\beta \gamma + \beta \xi}{1+\alpha \xi } \right) y - c\xi y^{2}.
\end{equation*}

Now comparison with a logistic ODE gives the boundedness of $y$. The equation for the pest $x$ is unchanged, so boundedness of $x$ follows also via comparison to a logistic ODE.

\end{proof}

\section{Equilibrium Analysis}
\label{sec:app}

\subsection{Linear Stability}

The Jacobian matrix $\left(J\right)$ of the system ~\eqref{Eqn:1nn} can be represented as\\
\begin{align}\label{Eqn:JacobianMain}
J=
\begin{bmatrix}
J_{11} &  J_{12}\\
J_{21} & J_{22}
\end{bmatrix},
\end{align}
where $J_{11}=1-\frac{2x^*}{\gamma}-\frac{\left(1+\alpha\xi \right)y^*}{\left(1+\alpha\xi+x^*\right)^2},~ J_{12}= -\frac{x^*}{\left(1+\alpha\xi+x^*\right)},~ J_{21}=\frac{\left(1+\alpha\xi-\xi\right)\beta y^*}{\left(1+\alpha\xi+x^*\right)^2}$ and $J_{22}=\frac{\beta\left(x^*+\xi\right) }{\left(1+\alpha\xi+x^*\right)}-\delta-2c\xi y^*.$ \\

\noindent Then the trace (tr) and determinant (det) of $J$ is given by

\begin{align}\label{equa1}
    \Tr(J)=& J_{11}+J_{22}  \nonumber\\
    =& 1-\frac{2x^*}{\gamma}-\frac{\left(1+\alpha\xi \right)y*}{\left(1+\alpha\xi+x^*\right)^2}+\frac{\beta\left(x^*+\xi\right) }{\left(1+\alpha\xi+x^*\right)}-\delta-2c\xi y^* \nonumber\\
     =& \frac{x^*}{\left(1+\alpha\xi+x^*\right)}\left[ 1-\frac{1}{\gamma}\left(1+\alpha\xi+2x^*\right)\right]-c\xi y^*. 
\end{align}

\begin{align}\label{equa2}
    \det(J)=& J_{11}J_{22}-J_{12}J_{21} \\
    =& \left[ 1-\frac{2x^*}{\gamma}-\frac{\left(1+\alpha\xi \right)y*}{\left(1+\alpha\xi+x^*\right)^2} \right]\cdot\left[\frac{\beta\left(x^*+\xi\right) }{\left(1+\alpha\xi+x^*\right)}-\delta-2c\xi y^*\right] \nonumber\\
     -& \left[\frac{-x^*}{\left(1+\alpha\xi+x^*\right)^2} \right]\cdot\left[\frac{\left(1+\alpha\xi-\xi\right)\beta y^*}{\left(1+\alpha\xi+x^*\right)^2}\right]. \nonumber \\
\end{align}
\subsubsection{\bf{Extinction state $(0,0)$}}

\begin{lemma}
\label{lem:lems1s0}
Consider $\xi> \frac{\delta}{\beta-\delta\alpha}$, then the extinction state $(0, 0)$ is unstable as a source.
\end{lemma}
\begin{proof}
The trace at the extinction state is given by
$ \Tr(J) = 1 + \frac{\beta \xi}{1 + \alpha \xi} - \delta $, thus if $\xi> \frac{\delta}{\beta-\delta\alpha}$, then  $ \Tr(J) > 0, \det(J) > 0$, thus $(0,0)$ is unstable.
\end{proof}
\subsubsection{\bf{Pest free state $(0,y^*)$}}

\begin{lemma}
\label{lem:lems1s}
Consider $\xi> \frac{\delta}{\beta-\delta\alpha}$, then the pest free state $(0, y^{*})$ is a stable node if 
$ c<\frac{(\beta-\alpha\delta)\xi- \delta}{\xi(1+\alpha\xi)^2}$, and is a saddle if $ c > \frac{(\beta-\alpha\delta)\xi- \delta}{\xi(1+\alpha\xi)^2}$.
\end{lemma}

\begin{proof}
We substitute  $x^*=0$  and $y^*=\frac{(\beta -\delta \alpha)\xi-\delta}{c\xi(1+\alpha \xi)}$ into the equations  ~\eqref{equa1} and  ~\eqref{equa2}, to obtain,
$ \Tr(J) = -c \xi y^{*} < 0, \ \det(J) = - \left( 1-\frac{y^{*}}{1+\alpha \xi} \right)y^{*}$.
Via standard theory $(0,y^{*})$ is locally stable when,
$y^{*} = \frac{(\beta -\delta \alpha)\xi-\delta}{c\xi(1+\alpha \xi)} > 1+ \alpha \xi$, as then, $\det(J)  > 0$, and when $y^{*} = \frac{(\beta -\delta \alpha)\xi-\delta}{c\xi(1+\alpha \xi)} < 1+ \alpha \xi$, $\det(J)  < 0$, so we have a saddle.
\end{proof}

\subsubsection{\bf Predator free state $(\gamma,0)$}

\begin{lemma}
\label{lem:lems1}
Consider $\xi> \frac{\delta}{\beta-\delta\alpha}$, then the predator free state $(\gamma, 0)$ is a saddle.

\end{lemma}

\begin{proof}
We substitute $x^*=\gamma$ and $y^*=0$ into the equations  ~\eqref{equa1} and  ~\eqref{equa2}, and obtain \\

    $\Tr(J)|_{(\gamma,0)}=1-\frac{2\gamma}{\gamma}+\frac{\beta\left(\gamma+\xi\right) }{\left(1+\alpha\xi+\gamma\right)}-\delta=-1+\frac{\beta\left(\gamma+\xi\right) }{\left(1+\alpha\xi+\gamma\right)}-\delta$,\\\\

$ \det(J)|_{(\gamma,0)}=\left[ 1-\frac{2\gamma}{\gamma}\right]\cdot\left[\frac{\beta\left(\gamma+\xi\right) }{\left(1+\alpha\xi+\gamma\right)}-\delta \right]=-\frac{\left[\beta\left(\gamma+\xi\right)-\delta\left(1+\alpha\xi+\gamma\right)\right] }{\left(1+\alpha\xi+\gamma\right)}.  $\\\\

Thus when $\xi> \frac{\delta}{\beta-\delta\alpha}$ then $-\frac{(\beta-\delta\alpha)\xi-\delta}{\beta-\delta}<0$, yielding a negative determinant, and the result.
\end{proof}

\subsubsection{\bf{Coexistence equilibrium point  $(x^*,y^*)$ }}

To find the coexistence equilibrium point(s), of \eqref{Eqn:1nn}, we set the predator and prey nullcline equal to yield, 
 
 \begin{equation*}
 (1-\frac{x}{\gamma})(1+\alpha \xi+x)=\frac{\beta (x+\xi)}{c\xi(1+\alpha \xi+x)}-\frac{\delta}{c\xi}. 
 \end{equation*}

After some algebra, this yields the cubic equation,
%
%
%
%
%
%

\begin{align*}
 c \xi x^3 +[2c \xi (1+\alpha \xi) - c \xi \gamma]x^2 + [c\xi (1+\alpha \xi)^2 -2\gamma c\xi(1+\alpha \xi) + \gamma(\beta - \delta)]x \\
  +[\beta \gamma \xi - \gamma c \xi (1+\alpha \xi)^2 - \delta \gamma (1+\alpha \xi)]=0.
\end{align*}

\noindent The canonical form for the above is given by $x^3+ax^2+bx+d=0$, where \\
$a=2(1+\alpha \xi) -\gamma$,\\
$b=(1+\alpha \xi)^2 - 2\gamma (1+\alpha \xi) + \frac{\gamma}{c \xi}(\beta - \delta)$,\\
$d= \frac{\beta \gamma}{c} - \gamma (1+\alpha \xi)^2 - \frac{\delta \gamma}{c \xi}(1+\alpha \xi)$.

This third order algebraic equation derived from the nullclines can have one, two or three real roots.
Using the transformation $x=y-\frac{a}{3}$ we get the form to use Cardano's Formula $y^3+py+q=0$ where $p=b-\frac{a^3}{3}$ and $q=\frac{2a^3}{27}-\frac{ab}{3}+d$.

Simplifying for $p$ and $q$ we get,
$p= -\frac{(1+\alpha \xi)^3}{3} - \frac{2\gamma (1+\alpha \xi)}{3} + \frac{\gamma (\beta - \delta)}{c \xi}-\frac{\gamma^2}{3}$ and\\
$q=-\frac{2(1+\alpha \xi)^3}{27}+\frac{10\gamma(1+\alpha \xi)^2}{27}+\gamma(1+\alpha \xi)[-\frac{14\gamma}{27}-\frac{(2\beta+\delta)}{3c\xi}]+\frac{\gamma^2(\beta-\delta)}{3c\xi}-\frac{2\gamma^3}{27}+\frac{\beta \gamma}{c}$.\\

The discriminant is a function of all the parameters in the model where,
$\Delta=(\frac{q}{2})^2+(\frac{p}{3})^3$. Thus, the following roots can be obtained 
\begin{enumerate}
    \item[(i)] $x_1=A+B -\frac{a}{3}$
    \item[(ii)] $x_2=-\frac{A+B}{2} -\frac{a}{3} + \frac{i\sqrt{3}}{2}(A-B)$
    \item[(iii)] $x_3=-\frac{A+B}{2} -\frac{a}{3} - \frac{i\sqrt{3}}{2}(A-B)$
\end{enumerate}

where $A=\sqrt[3]{-\frac{q}{2}+\sqrt{\Delta}}$ and $B=\sqrt[3]{-\frac{q}{2}-\sqrt{\Delta}}$.

\begin{theorem}
 If $\beta(\gamma+\xi)\leq \delta(1+\alpha\xi)$, then the model \eqref{Eqn:1nn} does not have an interior equilibrium. If $\beta(\gamma+\xi) > \delta(1+\alpha\xi)$, then the model \eqref{Eqn:1nn} has at least one positive real equilibrium, or at most three positive equilibria. The region can be separated as
 \begin{enumerate}
     \item Suppose $\Delta > 0$, 
     \begin{enumerate}
         \item If $(A+B)>\frac{a}{3}$ and $q<0$, then we have a unique positive equilibrium $A_1=(x_1,y_1)$ where $x_1=A+B-\frac{a}{3}$.
         \item If $(A+B)>\frac{a}{3}$, $q>0$ and $p<0$, then we have a unique positive equilibrium  $A_1=(x_1,y_1)$ where $x_1=A+B-\frac{a}{3}$.
     \end{enumerate}
     \item Suppose $\Delta < 0$ and $p\neq 0$, then we have three unique equilibria.  
     \item Suppose $\Delta = 0$, 
     \begin{enumerate}
         \item If $q<0$ and $a>0$, then we have positive unique equilibria $A^*=(x^*,y^*)$ where $x^*=2A-\frac{a}{3}$.
        \item If $q>0$, $a<0$ and $A>\frac{a}{6}$, then we have two positive roots, a single root $A^*=(x^*,y^*)$ and a double root $A_{2,3}=(x_{2,3},y_{2,3})$  where $x^*=2A-\frac{a}{3}$ and $x_{2,3}=-A-\frac{a}{3}$.
        \item If $q<0$, $a<0$ and $-A<\frac{a}{3}$, then we have a positive unique equilibria $A^*=(x^*,y^*)$ where $x^*=2A-\frac{a}{3}$.
        \item If $q<0$, $a<0$ and $-A>\frac{a}{3}$, then we have two positive roots, a single root $A^*=(x^*,y^*)$ and a double root $A_{2,3}=(x_{2,3},y_{2,3})$  where $x^*=2A-\frac{a}{3}$ and $x_{2,3}=-A-\frac{a}{3}$.
        \item If $q=0$, then we have a unique equilibrium $A_1=(x_1,y_1)$ where $x_1=-a/3$.
     \end{enumerate}
     \end{enumerate}
\end{theorem}

%
%
%
%

\section{Dynamics of the additional food model with competition}
We will consider the effect of this competition on the predator pest dynamics in both the small and large additional food regimes.

%
%
%
%

\subsection{Concave down predator nullcline}
Note, the prey nullcline maintains the same shape as earlier, however the predator nullcline could be concave up or down, depending on whether the pest extinction state $(0,y^*)$ is above the horizontal asymptote $y=\frac{\beta-\delta}{c\xi}$, of the predator nullcline, or below. That is, $y^* > \frac{\beta-\delta}{c\xi} \Leftrightarrow \xi>\frac{ 1}{1-\alpha}$, or
$y^* < \frac{\beta-\delta}{c\xi} \Leftrightarrow \xi < \frac{ 1}{1-\alpha}$. 
We consider the latter case first, and note first that the pest free state $(0, y^*) $ is given by, setting $x=0$, in \eqref{Eqn:1nn}, to yield,
 \begin{equation*}
   y^*=\frac{1}{c\xi}\left(\frac{\beta \xi}{1+\alpha \xi}-\delta  \right) > 0,
  \end{equation*}
 
 since we restrict $\xi>\frac{\delta}{\beta-\delta \alpha}$.

\subsubsection{Persistence}
\begin{definition}
The system \eqref{Eqn:1nn} is said to be persistent if,
\begin{equation}
\liminf_{t \rightarrow \infty} x(t) > 0, \ \liminf_{t \rightarrow \infty} y(t) > 0.
\end{equation}

\end{definition}

\begin{definition}
The system \eqref{Eqn:1nn} is said to be \emph{uniformly} persistent if there exists a positive number $\epsilon$ s.t.,

\begin{equation}
\liminf_{t \rightarrow \infty} x(t) \geq \epsilon > 0, \ \liminf_{t \rightarrow \infty} y(t) \geq \epsilon > 0.
\end{equation}

\end{definition}

We now state some persistence type results for the $y$ component.

\begin{lemma}
\label{lem:ylb}
Consider the predator-pest system described via \eqref{Eqn:1nn}, and assume the parametric restrictions, $\beta - \delta \alpha >0$,
$\xi > \frac{\delta}{\beta - \delta \alpha}$, then for all solutions $(x(t), y(t))$ initiating from any positive initial condition $(x_{0}, y_{0})$, $y$ is  persistent.
\end{lemma}

\begin{proof}

Note the non-negativity of solutions, and the instability of the $(0,0)$ and the 
$(\gamma, 0)$ states, see Lemma \ref{lem:lems1s}, Lemma \ref{lem:lems1}, yield,

\begin{equation*}
\liminf_{t \rightarrow \infty} y(t)  > 0.
\end{equation*}

The result follows.
\end{proof}

\begin{lemma}
\label{lem:ylb}
Consider the predator-pest system described via \eqref{Eqn:1nn}, and assume the parametric restrictions, $\beta - \delta \alpha >0$,
$\xi > \frac{\delta}{\beta - \delta \alpha}$, then for all solutions $(x(t),y(t))$ initiating from any positive initial $x_{0}$, and $y_{0} > \frac{1}{c \xi} \left( \frac{\beta \xi}{1 + \alpha \xi} - \delta \right)$, $y$ is uniformly persistent.

\end{lemma}

\begin{proof}
We have from the equation for $y$, positivity of solutions, and the boundedness of the pest $x$ by its carrying capacity $\gamma$ that,
\begin{equation*}
 \frac{dy}{dt} = \frac{\beta xy}{1+\alpha \xi + x} + \frac{\beta \xi y}{1+\alpha \xi + x} - \delta y -c\xi y^{2}  
>   \left( \frac{\beta \xi }{1+\alpha \xi } - \delta \right) y -c\xi y^{2}. 
 \end{equation*}

This follows from the monotone increasing property of the functional response, $ \frac{\beta x + \beta \xi}{1+\alpha \xi + x}$ on $x \in [0,\infty)$. Now comparison with a logistic ODE yields 

\begin{equation*}
y \geq \frac{1}{c \xi} \left( \frac{\beta \xi}{1 + \alpha \xi} - \delta \right),
\end{equation*}

the lower boundedness of $y$, given that we start above this bound initially, so if $y_{0} \geq \frac{1}{c \xi} \left( \frac{\beta \xi}{1 + \alpha \xi} - \delta \right)$. Thus we have that for any $\epsilon$ chosen s.t.,
$0 < \epsilon < \frac{1}{c \xi} \left( \frac{\beta \xi}{1 + \alpha \xi} - \delta \right)$,

\begin{equation*}
\liminf_{t \rightarrow \infty} y(t) \geq \frac{1}{c \xi} \left( \frac{\beta \xi}{1 + \alpha \xi} - \delta \right) > \epsilon > 0.
\end{equation*}

This proves the lemma.
\end{proof}

\subsubsection*{Bi-stability}

Since the persistence of the $y$ component has been established we classify the possible dynamics into various cases depending mainly on the $\xi-c$ parameter space.
We state the following theorem,

\begin{theorem}
\label{thm:1o3}
Consider the parametric restriction $ \frac{\delta}{\beta -\delta \alpha} < \xi < \frac{ 1}{1-\alpha}$. Suppose that,

\begin{equation*}
     c > c^{*}_{3} = \frac{(\beta -\delta \alpha)\xi-\delta}{\xi(1+\alpha \xi)^2},
\end{equation*}

then $(0,y^*)$ exists as a saddle, with one or three interior equilibrium. If

\begin{equation*}
   \frac{4\gamma(\beta-\delta)}{\xi \left(1+\alpha \xi+\gamma \right)^2} = c^{*}_{2} < c < c^{*}_{3} = \frac{(\beta -\delta \alpha)\xi-\delta}{\xi(1+\alpha \xi)^2}.
\end{equation*}

 Then $(0,y^*)$ is locally stable, and there are two interior equilibrium, in which case bi-stability is a possibility. 
\end{theorem}

\begin{proof}

We can find the maximum of the prey nullcline by taking the derivative of the prey nullcline, and setting it to zero,
$\frac{dy(x)}{dx}=1- \frac{(1+\alpha \xi)}{\gamma}-\frac{2x}{\gamma} = 0 \Leftrightarrow  x=\frac{1}{2}\left[\gamma- 1-\alpha\xi\right]$, and plugging this value of $x$ into the prey nullcline, yields the maximum of the prey nullcline  $y_{max}(x)=\left(1-\frac{\gamma- 1-\alpha\xi}{2\gamma}\right) \left(1+\alpha \xi+\frac{\gamma- 1-\alpha\xi}{2}\right)=\frac{\left(1+\alpha \xi+\gamma\right)^2}{4\gamma}$. Note the horizontal asymptote of predator nullcline is given by,

$\lim_{x \to\infty} \frac{\beta (1+\frac{\xi}{x})}{c\xi\left[\frac{(1+\alpha \xi)}{x}+1\right]} -\frac{-\delta}{c\xi}= \frac{\beta-\delta}{c\xi}$.

If we enforce that the maximum of the prey nullcline is above the horizontal asymptote of the predator nullcline, while  $y^*>1+\alpha\xi$, then this yields two interior equilibrium, with a pest free stable state. That is  $\frac{\left(1+\alpha \xi+\gamma\right)^2}{4\gamma}>\frac{\beta-\delta}{c\xi}$. Equivalently, this gives, 

\begin{equation*}
   \frac{4\gamma(\beta-\delta)}{\xi \left(1+\alpha \xi+\gamma \right)^2} < c < \frac{(\beta -\delta \alpha)\xi-\delta}{\xi(1+\alpha \xi)^2}.
\end{equation*}

Now if $y^{*} < 1+\alpha \xi$, we have that it is a saddle via lemma \ref{lem:lems1s}. This happens if,

\begin{equation*}
    c > \frac{4\gamma(\beta-\delta)}{\xi\left(1+\alpha \xi+\gamma\right)^2}.
\end{equation*}

\end{proof}

\begin{theorem}\label{thm:gloint}
Consider the parametric restriction $ \frac{\delta}{\beta -\delta \alpha} < \xi < \frac{ 1}{1-\alpha}$. One can choose $c$ sufficiently large, that is $c > \frac{1}{ 2 \xi (1+\alpha \xi) }$, s.t. there exists only one interior equilibrium, which is globally attracting.  

\end{theorem}

\begin{proof}

We first apply the Dulac criterion to \eqref{Eqn:1nn}. Consider the Auxilliary function. $\phi(x,y) = \frac{1}{x y}$, 

\begin{eqnarray}
&& \nabla \cdot (\phi(x,y) \frac{dx}{dt}, \phi(x,y)\frac{dy}{dt}) \nonumber \\
&=&\frac{ \partial}{\partial x} \left(\frac{1}{x y} \left( x(1-\frac{x}{\gamma}) - \frac{xy}{1+\alpha \xi + x}\right)\right)  \nonumber \\
 &+&  \frac{ \partial}{\partial y}\left(\frac{1}{x y}  \left( \frac{\beta xy}{1+\alpha \xi + x} + \frac{\beta \xi y}{1+\alpha \xi + x} - \delta y - c\xi y^{2}\right)  \right)    \nonumber \\
&=& -\frac{1}{\gamma y} + \frac{1}{(1+\alpha \xi + x)^{2}} - \frac{c \xi}{x} \nonumber \\
 & \leq &  \frac{1}{(1+\alpha \xi)^{2} + 2 (1+\alpha \xi)x + x^{2}} - \frac{c \xi}{x} \nonumber \\
 &\leq& \frac{1}{ 2 (1+\alpha \xi)x } - \frac{c \xi}{x} \nonumber \\
 &=& \frac{1}{x}\left( \frac{1}{2 (1+\alpha \xi)} - c \xi \right)  < 0 \nonumber \\
 \end{eqnarray}
 
\noindent as long as, $\frac{1}{ 2 (1+\alpha \xi) }- c \xi < 0$, or $c > \frac{1}{ 2 \xi (1+\alpha \xi) }$. This precludes the existence of limit cycles via the Dulac criterion.
Note, the extinction equilibrium $(0,0)$ is unstable, see section \ref{sec:app} the predator free equilibrium $(\gamma, 0)$ is a saddle, see lemma \ref{lem:lems1}, and the pest free equilibrium $(0, y^{*})$ is also a saddle under the parametric restrictions assumed, see lemma \ref{lem:lems1}. The stable manifolds of these are the predator and pest axis respectively, and so their $\omega-$limit sets cannot intersect in $\mathbb{R}^{2}_{+}$. Thus the solitary interior equilibrium is globally attracting.
\end{proof}

We now give sufficient conditions in terms of explicit parametric restrictions for global stability of $(0,y^*)$
\begin{theorem}\label{thm:glopestfree}

Consider the parametric restriction given by 
\begin{equation}
c < c^{*}_{1} = \min \left( \frac{4\gamma(\beta-\delta)}{\xi\left(1+\alpha \xi+\gamma\right)^2} ,  \frac{-\delta+(\beta -\delta \alpha)\xi}{\xi(1+\alpha \xi)^{2}} ,   \frac{4 \beta}{\xi} \left(\frac{1 + \alpha \xi - \xi}{1 + \alpha \xi + \gamma}\right) \left( \frac{\gamma}{\gamma -(1 + \alpha \xi )}\right) \right),
  \end{equation}

  then $(0,y^*)$ is globally stable, with no interior equilibrium.
\end{theorem}

\begin{proof}
We first require that $y^{*} > 1 + \alpha\xi$, else $(0,y^*)$ is a saddle. This is true if

\begin{equation}
\label{eq:1lk}
c <  \left( \frac{-\delta+(\beta -\delta \alpha)\xi}{\xi(1+\alpha \xi)^{2}} \right).
  \end{equation}
  
  We also require that the maximum of the prey nullcline is below the horizontal asymptote of predator nullcline. This happens if,
 
  \begin{equation}
  \label{eq:2lk}
 c < \frac{4\gamma(\beta - \delta)}{\xi\left(1+\alpha \xi+\gamma\right)^2}. 
    \end{equation}
 
 Next, if there are to be no intersections between the predator and prey nullclines, we require that the predator nullcline $g(x)$ remain higher than the prey nullcline $f(x)$ on the interval $x \in [0, \frac{1}{2}(\gamma -1 - \alpha \xi)]$. A sufficient condition for this is if $g(0) \geq f(0)$ and $\min\left( g^{'}(x)  \right) > \max \left( f^{'}(x)  \right)$ for $x \in [0, \frac{1}{2}(\gamma -1 - \alpha \xi)]$, since the minimum of $g^{'}(x)$ occurs at $x=\frac{1}{2}(\gamma -1 - \alpha \xi)$, and the maximum of 
  $f^{'}(x)$ is at $x=0$, computing the above yields,
  
  \begin{equation}
  \label{eq:3lk}
c < \frac{4 \beta}{\xi} \left(\frac{1 + \alpha \xi - \xi}{1 + \alpha \xi + \gamma}\right) \left( \frac{\gamma}{\gamma -(1 + \alpha \xi )}\right).
  \end{equation}
  
  Taking the minimum of the bounds for $c$ in \eqref{eq:1lk}, \eqref{eq:2lk} and \eqref{eq:3lk} yields that there is no interior equilibrium. Thus there exists $(0,0)$ which is unstable, and $(\gamma, 0)$ which is a saddle with its stable manifold as the x-axis (prey axis). This precludes the existence of homoclinic and heteroclinic connections. We next show there can be no periodic orbits.

\end{proof}

\subsection{Concave up predator nullcline}

We now turn to the case where the predator nullcline is concave up, that is when $\xi > \frac{1}{1-\alpha}$.

\begin{lemma}
 Consider the parametric restriction $\max\left(\frac{1}{1-\alpha}, \frac{\delta}{\beta -\delta \alpha}\right) < \xi$. Suppose that,

\begin{equation*}
 \frac{4\gamma(\beta-\delta)}{\xi\left(1+\alpha \xi+\gamma\right)^2} < c< \frac{(\beta -\delta \alpha)\xi-\delta}{\xi(1+\alpha \xi)^2}.
  \end{equation*}

Then there exists two interior equilibrium, with $(0,y^*)$ being locally stable. In this case bi-stability is a possibility. If

\begin{equation*}
 c > \frac{(\beta -\delta \alpha)\xi-\delta}{\xi(1+\alpha \xi)^2},
  \end{equation*}
then $(0,y^*)$ is a saddle, and there exists one interior equilibrium.

\end{lemma}

\begin{proof}
The result follows by enforcing $ 1+\alpha \xi < y^{*} < y_{max}$, where $y_{max}$ is the maximum value of the prey nullcline.
Also, if  $y^{*} < 1+\alpha \xi$, this enforces $(0,y^*)$ to be a saddle via lemma \ref{lem:lems1s}.
\end{proof}

We provide a result of sufficient conditions in terms of explicit parametric restrictions for global stability of $(0,y^*)$.

\begin{lemma}
 Consider the parametric restriction $\max\left(\frac{1}{1-\alpha}, \frac{\delta}{\beta -\delta \alpha}\right) < \xi$. Suppose that,

\begin{equation}
  c < \frac{4\gamma(\beta-\delta)}{\xi\left(1+\alpha \xi+\gamma\right)^2}. 
  \end{equation}

Then $(0,y^*)$ is globally stable.
\end{lemma}

\begin{proof}
The result follows by enforcing the horizontal asymptote of the predator nullcline is above $y_{max}$, the maximum value of the prey nullcline.
\end{proof}

\begin{figure}[!htb]
\begin{center}
    \includegraphics[scale=.457]{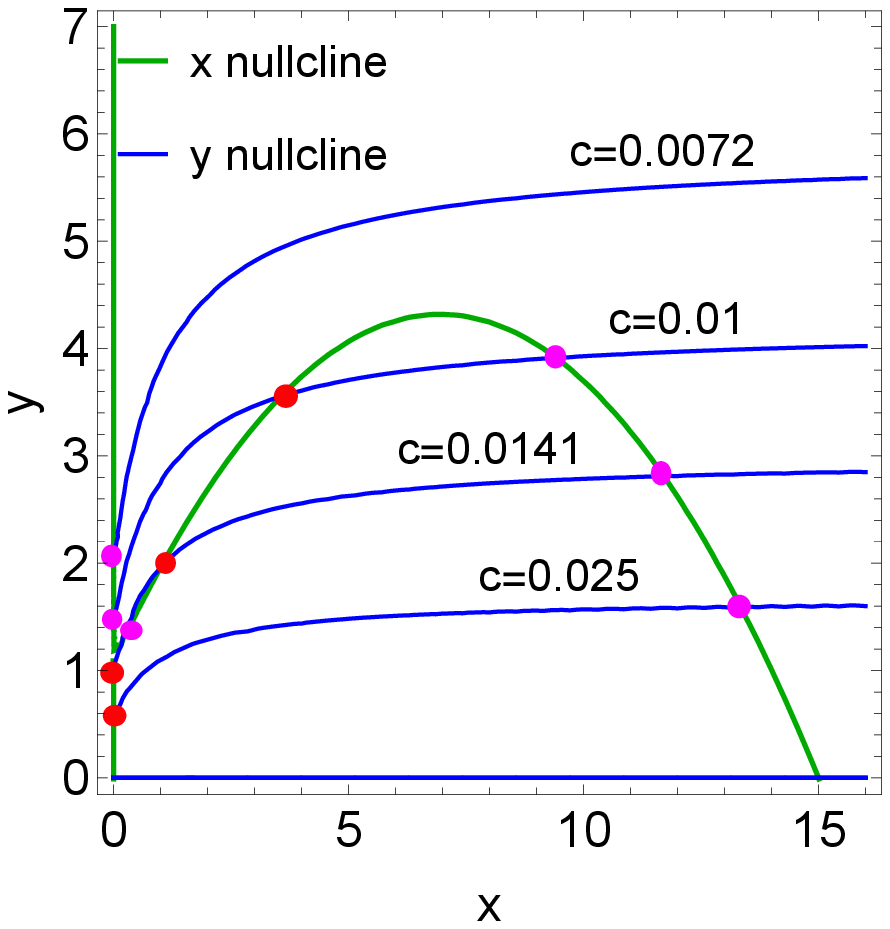}
    \includegraphics[scale=.467]{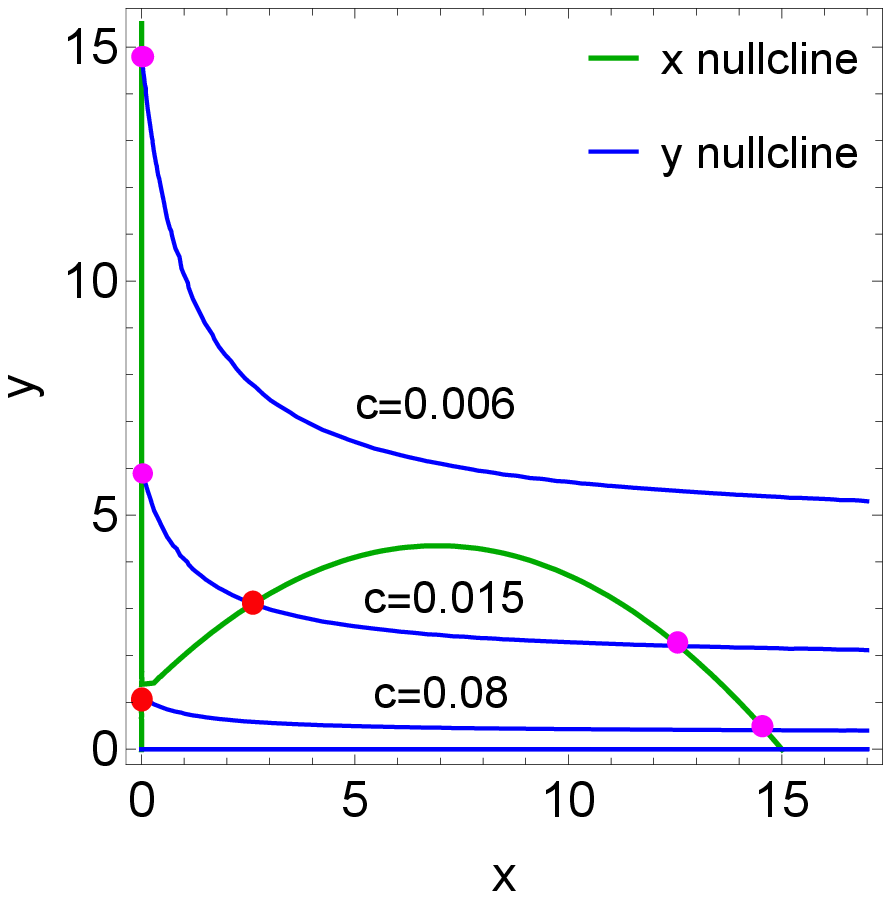}
    \includegraphics[scale=.467]{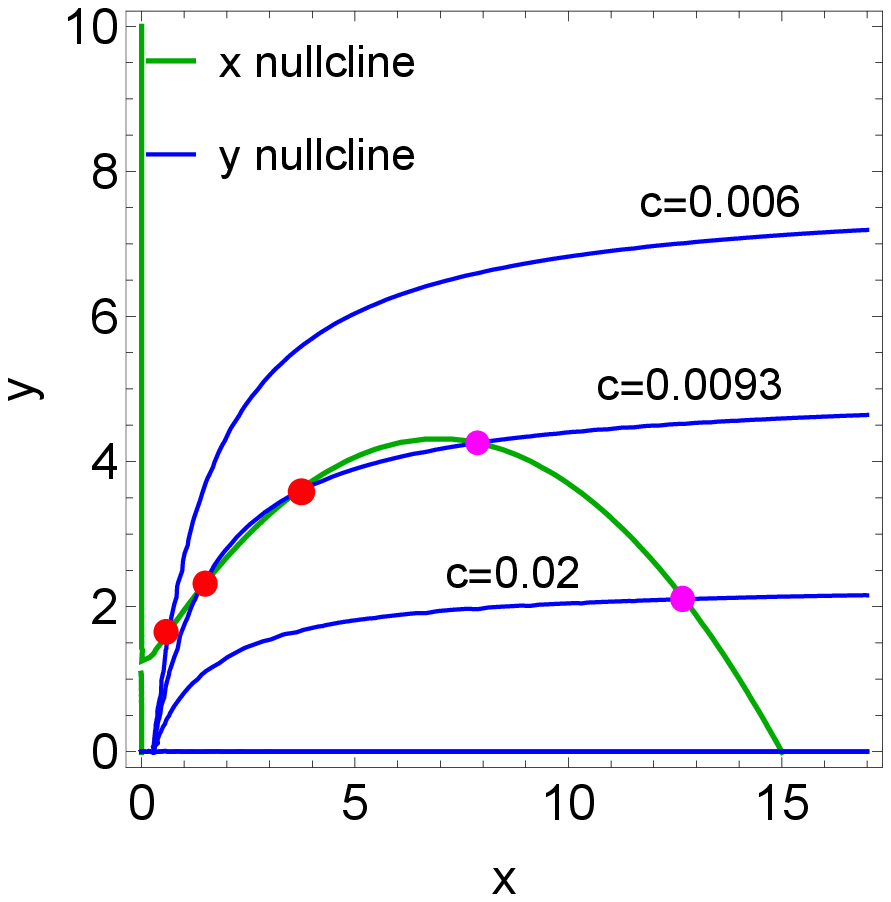}    
\end{center}
\caption{(a) (or left figure) Concave down predator nullcline for $\xi=1$ and varying $c$. (b) (or middle figure) Concave up predator nullcline for $\xi=1.5$ and varying c. (c) (or right figure) Biologically unfeasible pest extinction state for $\xi=0.9$ varying $c$. All other parameters are given by $\gamma=15, \alpha=0.1, \beta=0.3, \delta=0.258$. Solid magenta and red circles denote stable and unstable equilibrium points (pest free or interior) respectively.}
\label{fig:Nullclines}
\end{figure}

\subsubsection{Pest Extinction State with $\xi < \frac{\delta}{\beta - \delta \alpha}$}

We investigate the pest extinction state $(0, y^{*})$, in \eqref{Eqn:1nn}
under the parametric restriction $\xi < \frac{\delta}{\beta - \delta \alpha}$. This restriction enforces the predator nullcline to intersect the prey nullcline in the positive quadrant when $c=0$. In this case we know pest eradication is not achievable for any quantity of additional food, $\xi \in \left[0, \frac{\delta}{\beta - \delta \alpha}  \right)$. In that a positive interior equilibrium always occurs, and increasing the quantity of additional food can only lower the pest density - but the pest ultimately remains in the ecosystem. We state the following Lemma,

\begin{lemma}
\label{lem:l1p}
Consider the predator-pest system described via \eqref{Eqn:1nn}, and assume the parametric restriction 
$\xi < \frac{\delta}{\beta - \delta \alpha}$, then the pest extinction state $(0,y^{*})$, is biologically unfeasible.
\end{lemma}

\begin{proof}

The proof follows by solving for equilibrium when $x=0$, which yields

\begin{equation}
0 = \left(\frac{\beta \xi }{1+\alpha \xi } - \delta \right) y - c \xi y^{2}. 
\end{equation}

Since $\xi < \frac{\delta}{\beta - \delta \alpha} \Leftrightarrow \left( \frac{\beta \xi }{1+\alpha \xi } - \delta \right) < 0$,
and so $y^{*} = \frac{\left( \frac{\beta \xi }{1+\alpha \xi } - \delta \right) }{c \xi } < 0$. The result follows.

\end{proof}

\section{Bifurcations}

A goal of bifurcation theory is to classify qualitative changes in equilibrium as one or more control parameters are varied. When one traces the locus of these in two parameter space, the space is separated into regions of different qualitative behavior. In this section, we will provide numerical guidelines for co-dimension one and two bifurcations observed for \eqref{Eqn:1nn}. We have used the MATCONT package \cite{G05} in MATLAB R2022a and Mathematica 13 to produce the bifurcation plots.

\subsection{Co-dimension one bifurcations}

\subsubsection{Hopf and saddle-node bifurcation}
We first explore numerically the possibility of the occurrence of Hopf and saddle-node bifurcations around the interior equilibrium point $(x^*,y^*)$ with respect to varying the parameter $c$. These are clearly observed in Fig. \ref{fig:Hopf-SN-RegimesA} when the parameter $c$ crosses some critical thresholds. We delineate the bifurcation plot in Fig. \ref{fig:Hopf-SN-RegimesA}(a) into four regions i.e. $R_1, R_2, R_3$ and $R_4$. There are no interior equilibrium point in $R_1$ and the pest free $(0,y^*)$ is globally asymptotically stable, see proof in Theorem \ref{thm:glopestfree}. In $R_2$, we observe a saddle-node bifurcation at $c=c_{sn}=0.008836$ when $(x^*,y^*)=(6.51033,4.30728)$ and a subcritical Hopf bifurcation at  $c=c_{h}=0.008837$ when $(x^*,y^*)=(6.61665,4.31276)$. Two contrasting equilibrium points are observed in $R_3$. Additionally, in $R_4$, there exists three interior equilibrium points or a unique interior equilibrium point as $c$ is varied.  Also, we observed a saddle-node bifurcation at $c=c_{sn}=0.01507$ when $(x^*,y^*)=(0.41111,1.46969)$ and a subcritical Hopf bifurcation at $c=c_{h}=0.01359$ when $(x^*,y^*)=(0.01836,1.11699)$. For $c$ greater than the saddle-node point in $R_4$, the unique interior equilibrium point is globally asymptotically stable, see Theorem \ref{thm:gloint} for detailed proof.

\subsubsection{Transcritical bifurcation}

\begin{theorem}\label{thm:transcritical}
Consider \eqref{Eqn:1nn} when $\xi > \frac{1}{1 - \alpha}$, with $ \frac{4\gamma(\beta-\delta)}{\xi\left(1+\alpha \xi+\gamma\right)^2} < c< \frac{(\beta -\delta \alpha)\xi-\delta}{\xi(1+\alpha \xi)^2}$. Then  $c \nearrow  \frac{(\beta -\delta \alpha)\xi-\delta}{\xi(1+\alpha \xi)^2}$, results in a transcritical bifurcation, at $c= c^{*} = \frac{(\beta -\delta \alpha)\xi-\delta}{\xi(1+\alpha \xi)^2}$, in which the pest free state $(0,y^{*})$ changes stability, from a stable node to a saddle.
\end{theorem}
We present a different proof than is usually presented in the literature, which uses Sotoyomar's theorem \cite{K04}. Herein we are exploiting certain structural symmetries in our system.
\begin{proof}
We consider the area between the predator and prey nullclines, $y=f(x), y=g(x)$, from \eqref{Eqn:1nn} , see appendix.
\begin{equation}
\int^{x^{*}(c)}_{0}(f(x) - g(x))dx.
\end{equation}

The integral is considered, as $x$ varies from 0 to $x^{*}$, which is where the nullclines intersect. Note, $x^{*}$ has explicit dependence on the parameter $c$, and is a solution of the cubic,  $Ax^3+Bx^2+Cx+D=0$, where in particular,
 $D=c\xi\gamma(1+\alpha \xi)^2+\delta(1+\alpha \xi)\gamma-\beta\gamma\xi$, see appendix.
We now increase only the $c$ parameter, $ c \nearrow c^{*}$, and compute,

 \begin{equation}
\lim_{c \rightarrow c^{*}}\int^{x^{*}(c)}_{0}(f(x) - g(x))dx,
\end{equation}

where $c^{*} = \frac{(\beta -\delta \alpha)\xi-\delta}{\xi(1+\alpha \xi)^2}$, which is the transcritical bifurcation point. Now

 \begin{eqnarray}
&& 0 \leq \lim_{c \rightarrow c^{*}} \int^{x^{*}(c)}_{0}(f(x) - g(x))dx \nonumber \\
&& < \left(\lim_{c \rightarrow c^{*}} x^{*}(c)\right) |f(x) - g(x)| = 0. |f(x) - g(x)| = 0. \nonumber \\
\end{eqnarray}

This follows from continuity and the boundedness of the integral. The limit of $x^{*}(c)$ is evaluated directly from the expression for the cubic that $x^{*}$, is a solution, see .

 When $D=0$, 
 \begin{equation}
     \gamma\left[\left(1+\alpha \xi \right)^{2} c\xi-(\beta-\delta\alpha)\xi+\delta\right]=
 \frac{\gamma}{\xi \left(1+\alpha \xi \right)^{2}}\left[c-\frac{\left(\beta-\delta\alpha\right)\xi - \delta}{\xi \left(1+\alpha \xi \right)^{2}} \right]=0.
 \end{equation}

 This yields, 
  \begin{equation}
      c=\frac{\left(\beta-\delta\alpha\right)\xi - \delta}{\xi \left(1+\alpha \xi \right)^{2}}  = c^{*}.
 \end{equation}

Thus as $c \rightarrow c^{*}$, $D \rightarrow 0$, yielding $x^{*}=0$, and we have,

 \begin{equation}
\lim_{c \rightarrow c^{*}}\int^{x^{*}(c)}_{0}(f(x) - g(x))dx \rightarrow 0,
\end{equation}

implying that the saddle interior equilibrium collides with the pest free node, and this happens at the point of transcriticality $c=c^{*}$, after which the pest free state changes stability to a saddle, which follows directly from lemma \ref{lem:lems1s}, and the interior saddle equilibrium, now becomes a node (albeit negative, so is not biologically relevant) - this follows via standard theory for planar dynamical systems \cite{K04}, and the fact that the adjacent equilibrium, the pest free state is now a saddle so is unstable.
 This by definition is a transcritical bifurcation, and proves the theorem.

\end{proof}

\subsubsection{Pitchfork bifurcation}

\begin{theorem} 
\label{thm:pc1}
Consider \eqref{Eqn:1nn} when $\xi > \frac{1}{1 - \alpha}$, $\gamma < 1 + \alpha\xi$, 
$\frac{\beta \xi}{c\xi(1+\alpha \xi)}-\frac{\delta}{c\xi} < 1 + \alpha\xi$,
then a pitchfork bifurcation occurs as $c \searrow c^{*}$, where,
\begin{equation*}
    c^{*}=\frac{\beta\gamma (1+\alpha \xi-\xi)}{\xi(1+\alpha \xi)^2\left[\gamma-(1+\alpha\xi)\right]}.
\end{equation*}

\end{theorem}

\begin{proof}
The prey nullcline  is  given by $y=f(x)=(1-\frac{x}{\gamma})(1+\alpha \xi+x)=(1+\alpha \xi)+x-(1+\alpha \xi)\frac{x}{\gamma}-\frac{x^2}{\gamma}$, and the 
predator nullcline is given by $y=g(x)=\frac{\beta (x+\xi)}{c\xi(1+\alpha \xi+x)}-\frac{\delta}{c\xi} $.
So the slope of  prey nullcline at $x=0$ is
$\left.\frac{d f(x)}{dx}\right\vert_{x=0}=1- \frac{(1+\alpha \xi)}{\gamma}$. While the
slope of  predator nullcline at $x=0$ is
$\left.\frac{d g(x)}{dx}\right\vert_{x=0}=\frac{\beta (1+\alpha \xi-\xi)}{c\xi(1+\alpha \xi)^2}$.
Setting the two slopes equal,
$1- \frac{(1+\alpha \xi)}{\gamma}=\frac{\beta (1+\alpha \xi-\xi)}{c\xi(1+\alpha \xi)^2}$,
yields the following equation for $c^{*}$,
\begin{equation*}
    c^{*}=\frac{\beta\gamma (1+\alpha \xi-\xi)}{\xi(1+\alpha \xi)^2\left[\gamma-(1+\alpha\xi)\right]}.
\end{equation*}

Now as $c \searrow c^{*}$, the predator nullcline moves upwards, and the three interior equilibrium come closer together. This follows by the shape of the nullclines, under the parametric restriction enforced. Also, due to the parametric restrictions the pest free state is unstable as a saddle, as $g(0) < f(0)$. Since the two slopes of prey and predator nullcines at $x=0$ are chosen to be the same, by continuity, the three equilibrium merge into to one as $c \searrow c^{*}$, (see Theorem \ref{thm:transcritical}). Now as $c$ is further increased, the predator nullcline is completely above the prey nullcline, $g(x) > f(x)$ $\forall x$, and there is a single pest free equilibrium, which is stable. Thus by definition a pitchfork bifurcation has occurred.

\end{proof}

\subsection{Co-dimension two bifurcations}
Furthermore, we explore the possibility of the occurrence of some rich class of co-dimension two bifurcations in system \eqref{Eqn:1nn}. Since we are most interested in the interplay between the additional food and the predator competition, and the ecological and bio-control connotations, we choose to consider the two parameters $\xi$ and $c$ as our bifurcation parameters.

\subsubsection{Bogdanov-Takens  bifurcation}
Bogdanov-Takens (BT) bifurcation occurs for \eqref{Eqn:1nn}, when the Jacobian matrix at $(x^*,y^*)$ has a double zero eigenvalue in a two parameter plane. This is observed in Fig. \ref{fig:Hopf-SN-RegimesA}(b) at $(c,\xi)=(c_{bt},\xi_{bt})=(0.00882,1.02120)$ when $(x^*,y^*)=(6.60896,4.31360)$. 

\subsubsection{Cusp bifurcation}
A cusp bifurcation occurs when two saddle nodes collide. Typically for such a situation we need the 
existence of three interior equilibrium, where one is a saddle and two others are stable/unstable nodes. Here there is a possibility for one node to collide with the saddle along a certain path in the two parameter space, while along another path the other node collides with the saddle, setting up a collision
at the cusp point. This also becomes an organizing center for the considered model.
Typically in the neighborhood of a cusp one has a BT, and multiple (at least 3 equilibrium).
A cusp bifurcation occurs for \eqref{Eqn:1nn}, and is seen in Fig. \ref{fig:Hopf-SN-RegimesA}(b) at  $(c,\xi)=(c_{cp},\xi_{cp})=(0.00877,0.83313)$ when $(x^*,y^*)=(4.27778,3.83219)$.

\subsubsection{Saddle-node-transcritical bifurcation}
The saddle-node-transcritical bifurcation (SNTC), occurs when curves of saddle-node and transcritical bifurcations intersect. Amongst the first reported types of SNTC (in the predator prey literature) was in \cite{ BS07}, where a Lokta-Volterra predator-prey model with type I response was considered, and then in \cite{KSV10, KSV15}, where  predator competition was considered, as well as a constant term added to the prey, which (depending on its sign) can be interpreted as harvesting or migration. These are classified into the \emph{elliptic} case (single zero eigenvalue) or \emph{saddle} case (double zero eigenvalue). The SNTC is further complicated as it is not detected correctly by MATCONT, and only very recently has a numerical procedure been developed for its detection \cite{VH19}.

A saddle-node-transcritical (SNTC) bifurcation occurs for \eqref{Eqn:1nn}, and is demonstrated in Fig. \ref{fig:SNTC}. Here, one has the collision of a saddle-node (formed when the stable interior equilibrium $(x^*,y^*)$ and saddle $(x^{**},y^{**})$ collide) with the pest free state $(0,y^*)$, as it undergoes a transcritical bifurcation (which occurs when $y^* = 1 + \xi \alpha$), all in a single collision. 
We see that the single zero SNTC point occurred at $(\xi,c)=(\xi_{sntc},c_{sntc})=(1.01347,0.02155)$. This helps us to identify a critical organizing center of the system \eqref{Eqn:1nn}, when the parameters $\xi$ and $c$ are considered. For more details on the detection of saddle-node-transcritical interactions, please see \cite{BS07, VH19, V22} and references therein. We also observed BT bifurcation at $(\xi,c)=(\xi_{bt},c_{bt})=(1.01106,0.02109)$.

\subsubsection{Cusp-transcritical bifurcation}
A cusp-transcritical bifurcation (CPTC) occurs, when a cusp point collides with a transcritical bifurcation point . This bifurcation is demonstrated in Fig. \ref{fig:Cusp-Transcritical}. In Fig. \ref{fig:Cusp-Transcritical}(a), a one parameter bifurcation is constructed with $\xi$ as the bifurcation parameter. Locally, we observed the existence of a saddle-node and branch points. With the aid of MATCONT, the branch point (BP) was continued to produce transcritical bifurcation at $\xi=1.59411$ around the pest free state $(0,y^*)=(0,1.15941)$.  The saddle-node point (SN) at $\xi=1.59416$ around the coexistence equilibrium point $(x^*,y^*)=(0.00474,1.15598)$ was continued to produce the two parameter space diagram in  Fig. \ref{fig:Cusp-Transcritical}(b) as an additional parameter $c$ was varied. We observed a cusp point (CP) and two cusp-transcritical points ($CPTC_1$ and $CPTC_2$) of co-dimension two as $\xi$ and $c$ are varied. To the best of our knowledge this is the first time CPTC has occurred twice and in co-dimension two. In Fig. \ref{fig:Cusp-Transcritical}(c), we demonstrated that $CPTC_1$ and $CPTC_2$ lie on the transcritical curve $c(\xi)$ where $CPTC_1=TC_1=(\xi,c)=(1.52290,0.07778)$ and $CPTC_2=TC_2=(\xi,c)=(2.82631,0.11063)$. For more details on CPTC in co-dimension three, please refer to \cite{BS07,D20}.

\subsubsection{Pitchfork-transcritical bifurcation and an organising center}

When one traces the locus of bifurcation parameters in two parameter space, this separates regions of different qualitative behavior. Points in the parameter space where two or more bifurcation curves meet are called ``organizing centers" \cite{BSW07}. For generic dynamical centers the classification of the theory of these centers is nearly complete \cite{K04}. However, in systems with special structure bifurcations can occur in \emph{lower} co-dimension than the generic case. 

\begin{definition}[Pitchfork-transcritical bifurcation]
A pitchfork-transcritical bifurcation (PTC) occurs, when a pitchfork bifurcation point collides with a transcritical bifurcation point. 
\end{definition}

We first state the following result,

\begin{theorem} [Organizing center]
\label{thm:oc}
Consider \eqref{Eqn:1nn} when $\xi > \frac{1}{1 - \alpha}$, $\gamma < 1 + \alpha \xi$, and define  $a_{1}= (\beta - \delta \alpha) \alpha, \ b_{1}= -(\beta - \delta \alpha)(\gamma-1) + \beta \gamma (1-\alpha) + \delta \alpha, \ c_{1} = \beta \gamma - \delta + \delta \gamma$. Then if $b_{1}^{2}-4a_{1}c_{1}=0, \ \frac{-b_{1}}{2a_{1}} > 0$, there exists an organising center, in the $c-\xi$ parameter space, where a transcritical bifurcation point and pitchfork bifurcation point meet. 
\end{theorem}

\begin{proof}
By setting the slopes of the predator and prey nullclines equal when $x=0$, we obtain the following equation for $c$,
\begin{equation}
    c\left(\xi\right)=\frac{\beta\gamma (1+\alpha \xi-\xi)}{\xi(1+\alpha \xi)^2\left[\gamma-(1+\alpha\xi)\right]}.
\end{equation}

Now, note the TC curve

\begin{equation}
    c\left(\xi\right)= \frac{\left(\beta-\delta\alpha\right)\xi}{\xi \left(1+\alpha \xi \right)^{2}},
\end{equation}

setting these two expressions for $c(\xi)$ equal, yields a quadratic for $\xi$, 

\begin{equation}
a_{1} \xi^{2} + b_{1} \xi + c_{1} = 0, 
\end{equation}

where

\begin{equation}
    a_{1}= (\beta - \delta \alpha) \alpha, \ b_{1}= -(\beta - \delta \alpha)(\gamma-1) + \beta \gamma (1-\alpha) + \delta \alpha, \ c_{1} = \beta \gamma - \delta + \delta \gamma.
\end{equation}

Thus if $b_{1}^{2}-4a_{1}c_{1}=0, \ \frac{-b_{1}}{2a_{1}} > 0$, we have the occurrence of one positive root, $\xi^{*}$, using this we can compute $c^{*}=c(\xi^{*})$, which yields the organizing center $(c^{*},\xi^{*})$ in $c-\xi$ parameter space. This proves the theorem. 

\end{proof}

\begin{proposition}
\label{prop:pt2}
Consider \eqref{Eqn:1nn} when $\xi > \frac{1}{1 - \alpha}$, $\gamma < 1 + \alpha \xi$, $\frac{\beta \xi}{c\xi(1+\alpha \xi)}-\frac{\delta}{c\xi} < 1 + \alpha\xi$. A pitchfork-transcritical bifurcation can occur, in co-dimension two, where the bifurcating equilibrium has a single zero eigenvalue.
\end{proposition}

A numerical validation of Theorem \ref{thm:oc} and Proposition \ref{prop:pt2} is provided in Fig. \ref{fig:Pitchfork-Transcritical}.

\begin{remark}
We note that our construction of a pitchfork-transcritical bifurcation in Theorem \ref{thm:oc}, Proposition \ref{prop:pt2}, and seen numerically in Fig. \ref{fig:Pitchfork-Transcritical}, uses the symmetry of the predator nullcline about the y-(predator) axis, and the fact that we have a pest free state. However this symmetry requires one of the equilibrium points to be negative - thus this situation is not biologically feasible, but certainly possible mathematically. Similar symmetries have been exploited in our construction of the cusp-transcritical bifurcation seen numerically in Fig. \ref{fig:Cusp-Transcritical}.
\end{remark}

\subsection{Cyclical dynamics}

\subsubsection{Multiple limit cylces}
In Fig. \ref{fig:limit cycles and homoclinic}(a), we observe a large stable limit cycle containing all the three interior equilibrium points and a small unstable limit cycle containing one interior equilibrium point. The small unstable limit cycle is clearly seen in Fig. \ref{fig:limit cycles and homoclinic}(b).

\subsubsection{Homoclinic orbits}
A homoclinic orbit is created when the stable manifold $(W^S(x^*,y^*))$ and unstable  manifold $(W^U(x^*,y^*))$ coincide to form a loop. From our numerical simulations in Fig. \ref{fig:limit cycles and homoclinic}(c), a homoclinic loop is created at $c=c_{hom}=0.069358$.

\begin{figure}[!htb]
\begin{center}
    \includegraphics[scale=.2]{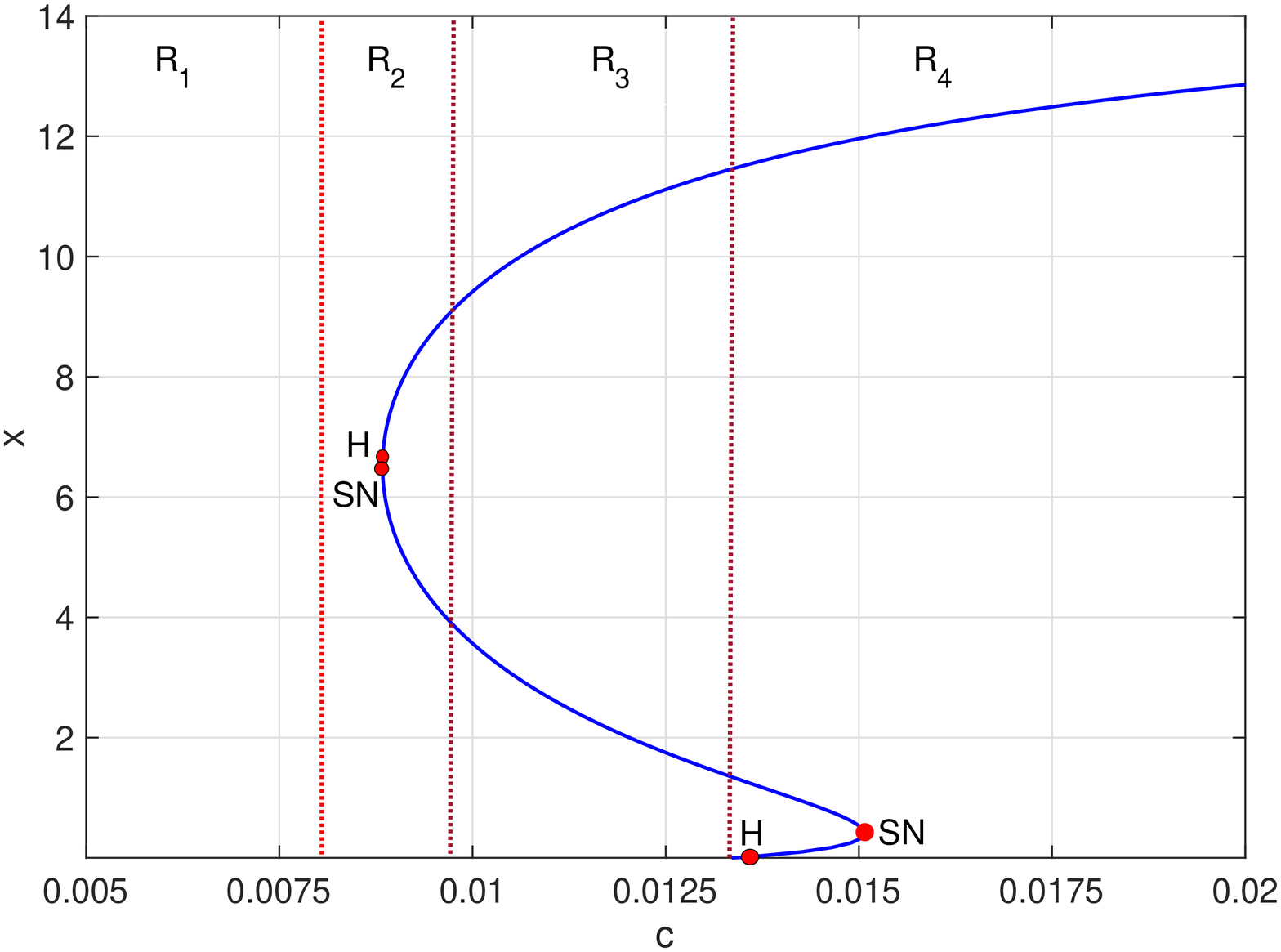}
   \includegraphics[scale=.2]{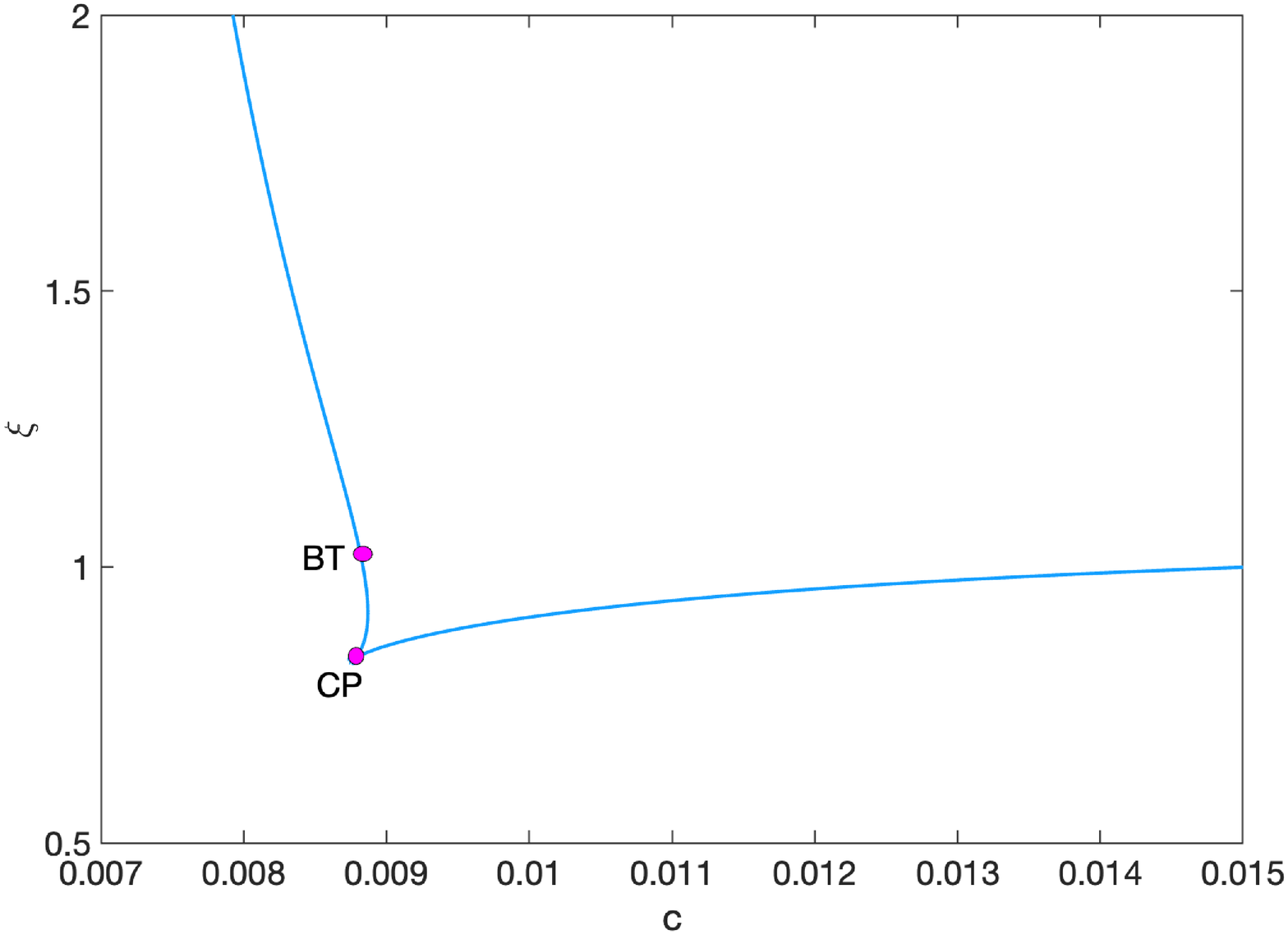}
\end{center}
\caption{(a) (or left figure)  Here we explore the effect of varying $c$, in the concave down predator nullcline. We observe the occurrence of Hopf and saddle-node bifurcations. Here $\gamma=15, \alpha=0.1, \xi=1, \beta=0.3, \delta=0.258$. (b) (or right figure) Two parameter bifurcation curve showing cusp and Bogdanov-Takens bifurcations in the $c-\xi$ plane. (Note: SN=Saddle-Node point, H=Hopf point, CP=Cusp point, BT=Bogdanov-Takens point.)}
\label{fig:Hopf-SN-RegimesA}
\end{figure}

\begin{figure}[!htb]
\begin{center}
    \includegraphics[width=5cm, height=5cm]{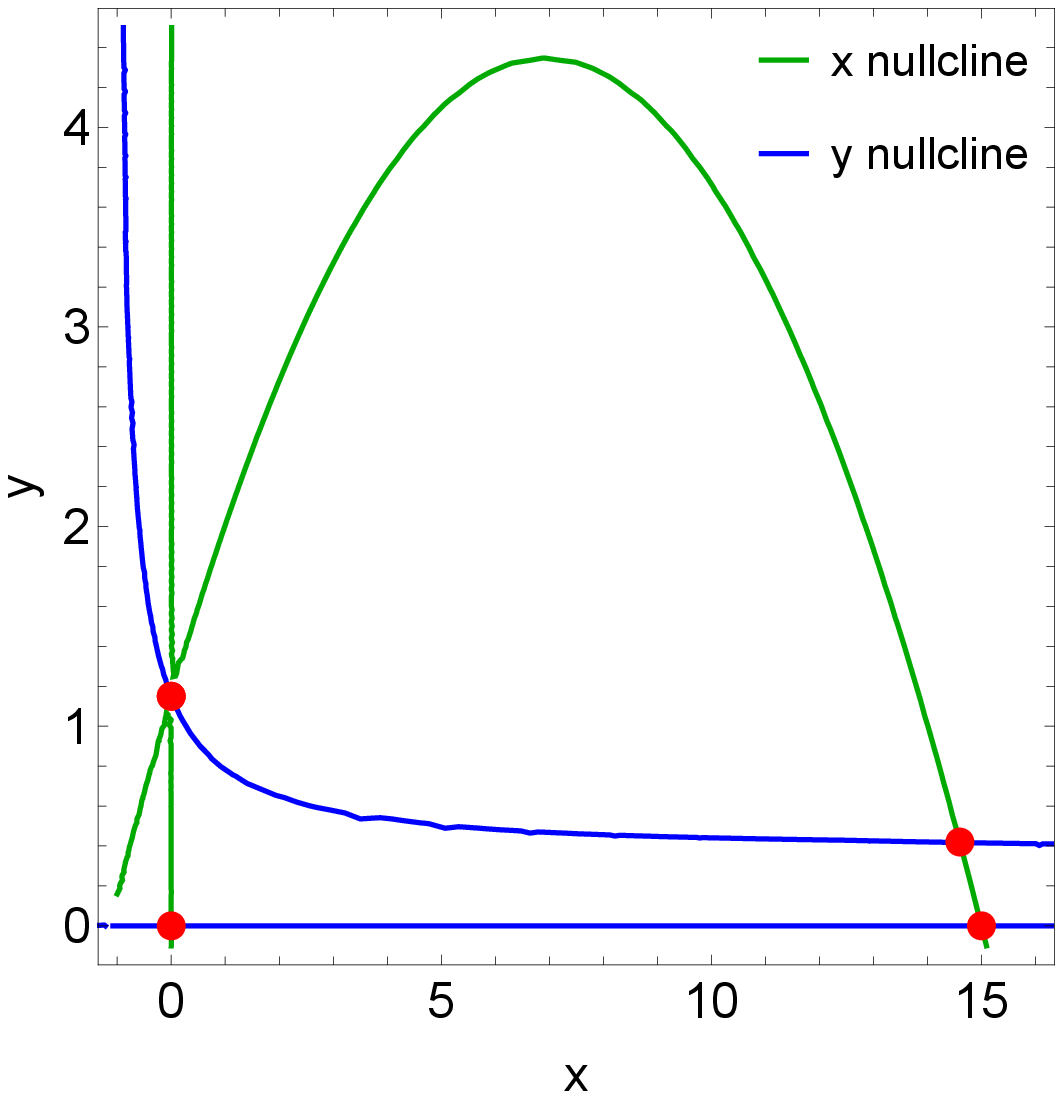}
    \includegraphics[width=5.22cm, height=5.5cm]{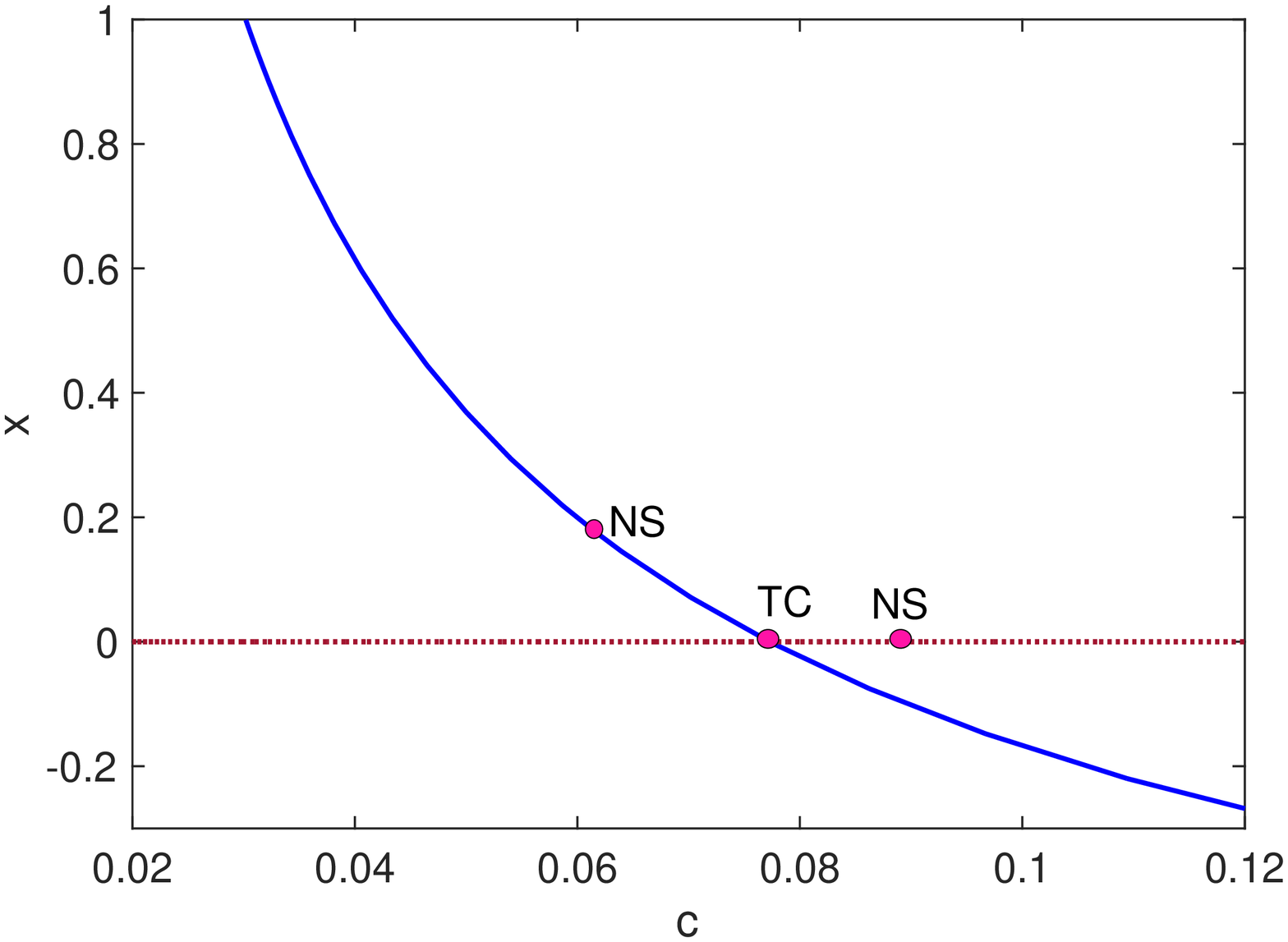}
    \includegraphics[width=6cm, height=5.3cm]{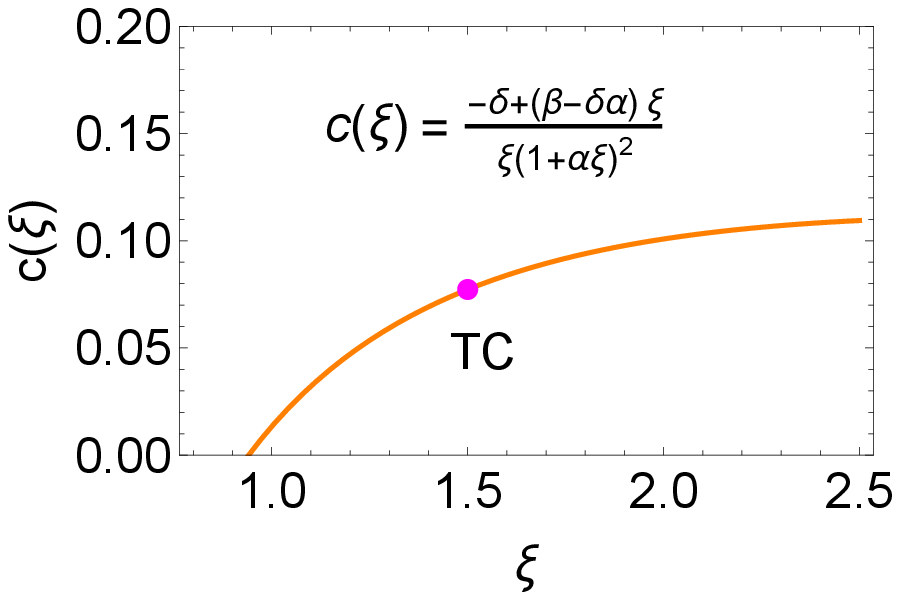}    
\end{center}
\caption{ Figures illustrating Theorem \ref{thm:transcritical} (a) (or top left figure) phase portrait, (b) (or top right) one parameter bifurcation diagram as $c$ is varied, and (c) (or bottom figure) transcritical point on the transcritical curve. The parameters are given by $\gamma=15, \alpha=0.1, \beta=0.3, \delta=0.258, \xi=1.5, c=c^*=0.077278$. (Note: NS=Neutral Saddle point, TC=Transcritical point.)}
\label{fig:Transcritical}
\end{figure}

\begin{figure}[!htb]
\begin{center}
    \includegraphics[scale=.28]{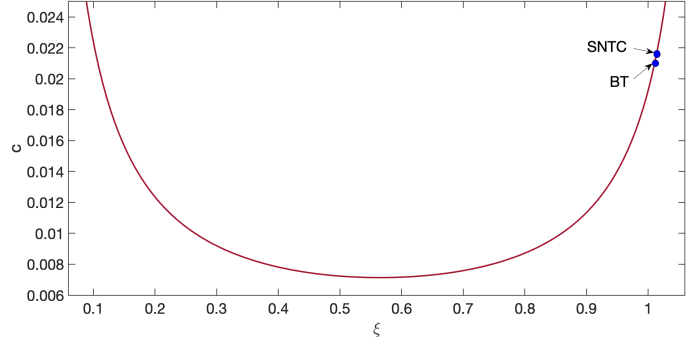}
\end{center}
\caption{Two parameter bifurcation curve in the $\xi - c$ plane. Here $\gamma=200, \alpha=0.10062, \beta=0.298, \delta=0.25$. In addition, we observe BT and SNTC points.  (Note: SNTC=Saddle-Node-Transcritical point, BT=Bogdanov-Takens point.)}
\label{fig:SNTC}
\end{figure}

\begin{figure}[!htb]
\begin{center}
    \includegraphics[scale=.25]{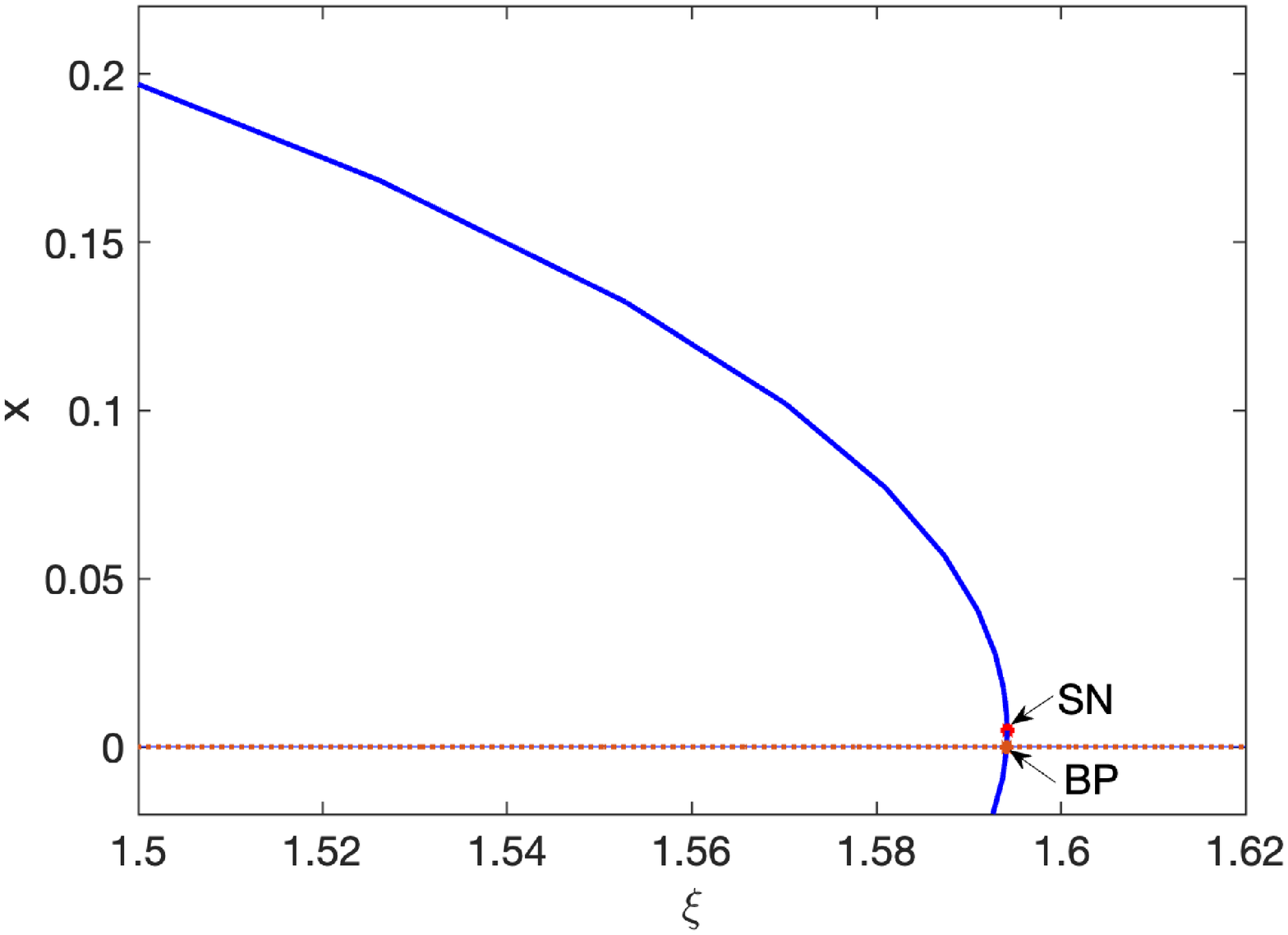}
    \includegraphics[scale=.25]{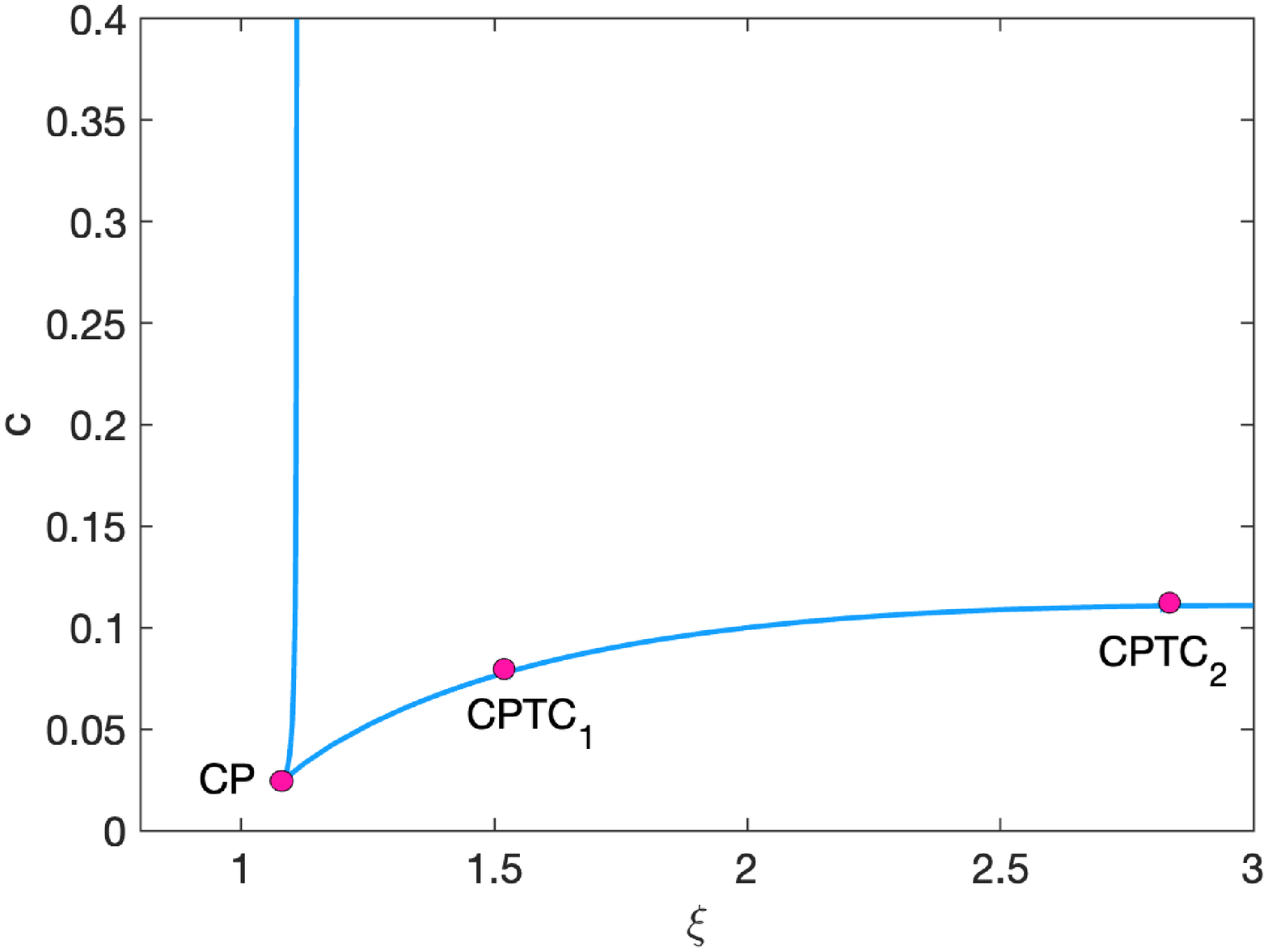}
    \includegraphics[scale=.34]{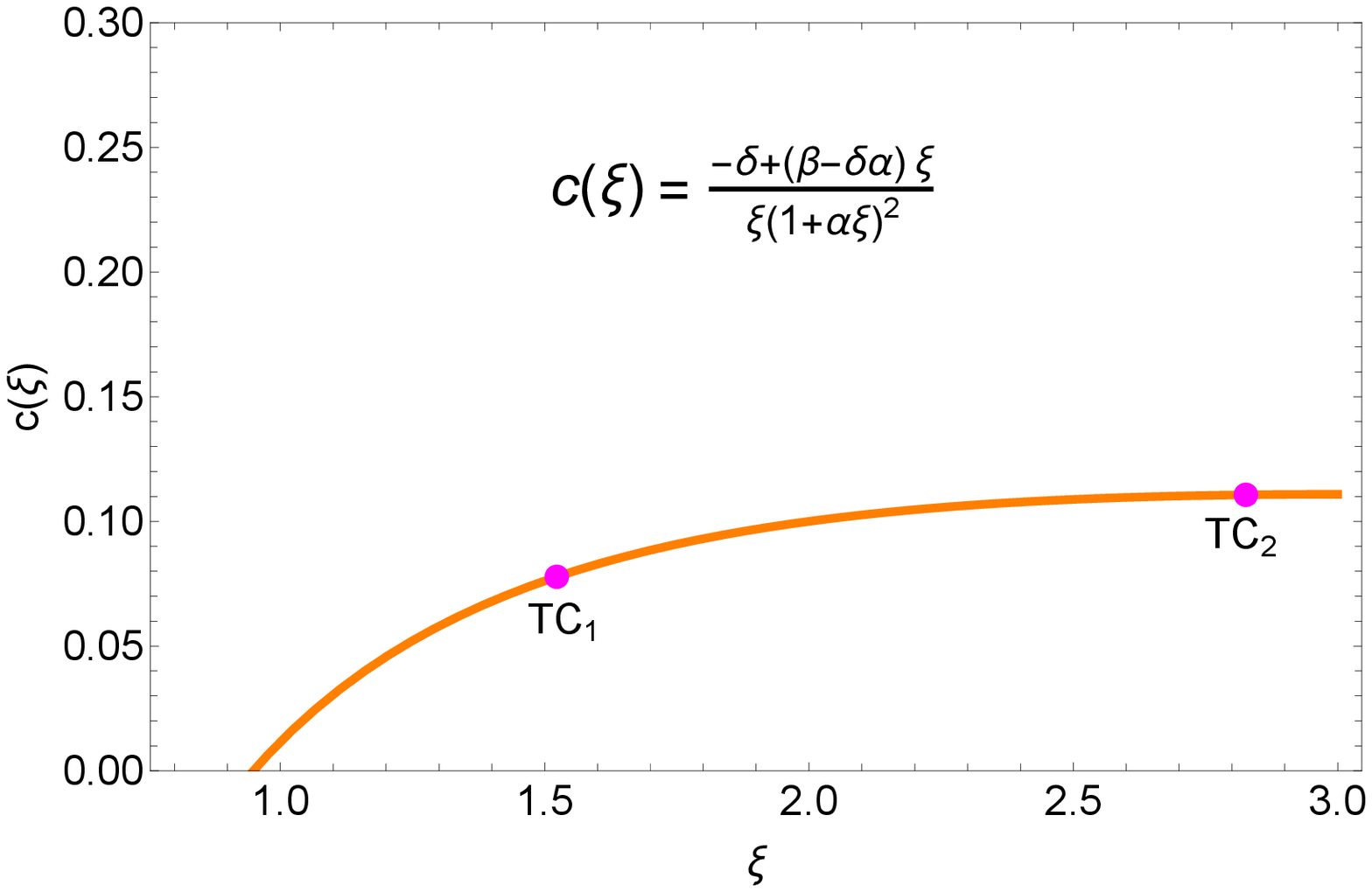}    
\end{center}
\caption{(a) (or top left figure) One parameter bifurcation diagram of pest with respect to the parameter $\xi$.  (b) (or top right figure) Two parameter bifurcation diagram in the $\xi-c$ parametric space. (c) (or bottom figure) Transcritical curve depicting transcritical points. Here $\gamma=0.675, \alpha=0.1, \beta=0.3, \delta=0.26$. (Note: CP=Cusp point, BP=Branch point, $CPTC_{i}$=Cusp-Transcritical point, $TC_{i}$=Trancritical point, $i=1,2$.)}
\label{fig:Cusp-Transcritical}
\end{figure}

\begin{figure}[!htb]
\begin{center}
    \includegraphics[width=5cm, height=4.5cm]{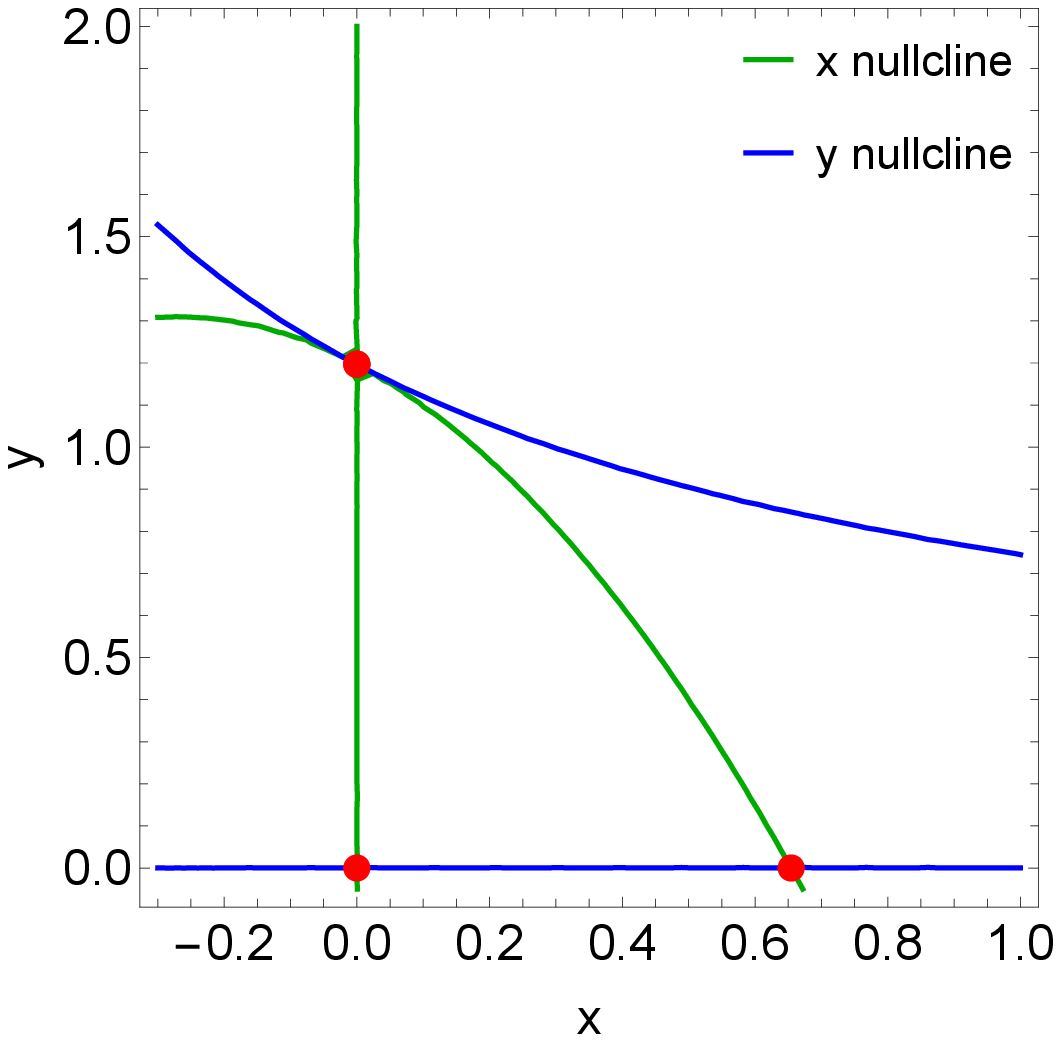}
    \includegraphics[width=5cm, height=4.5cm]{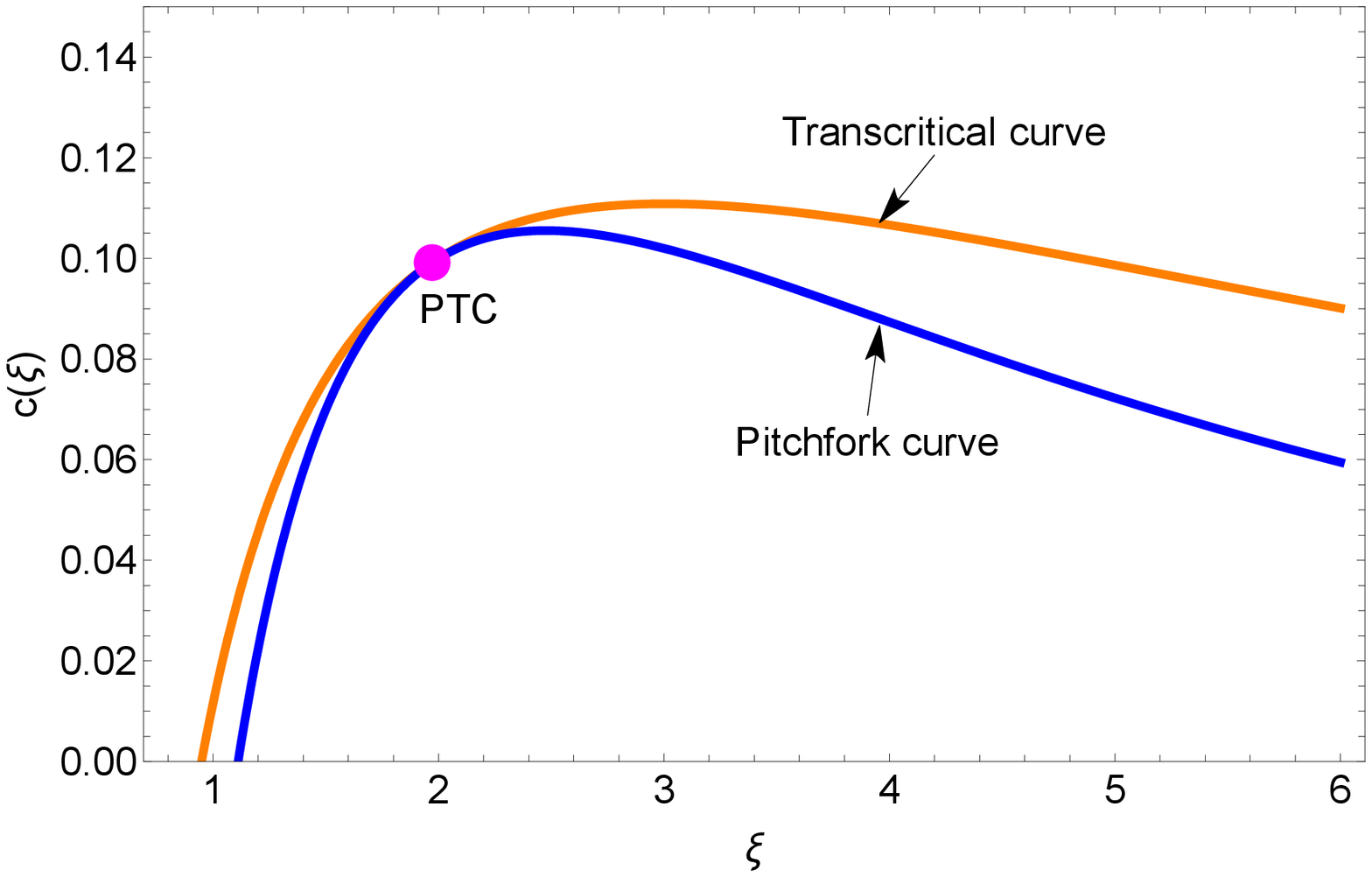}
    \includegraphics[scale=.3]{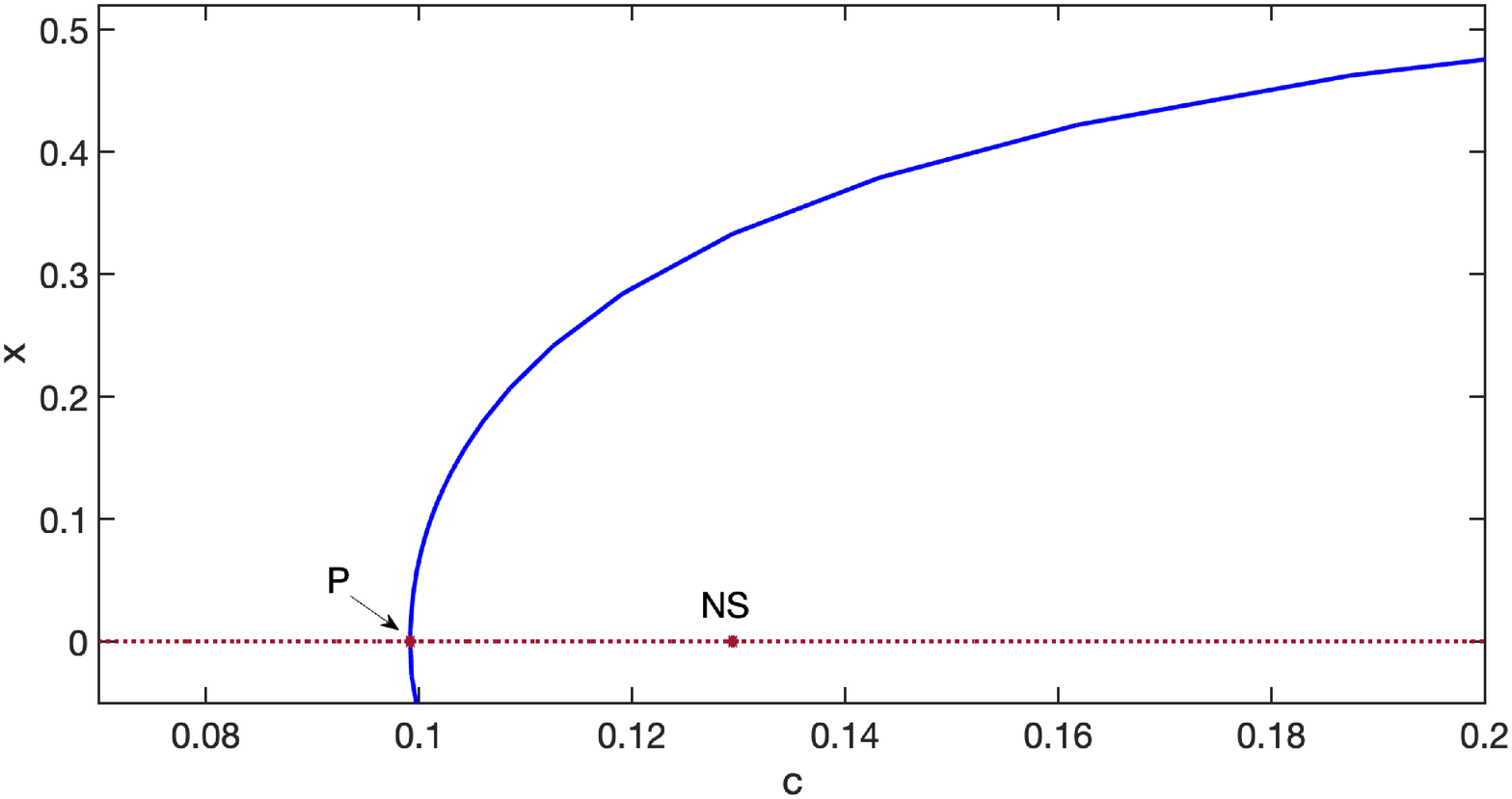}    
\end{center}
\caption{(a) (or top left figure) Phase portrait showing the instantaneous intersection of three equilibrium points at $(0,1.19717)$,  (b) (or top right figure) the intersection of pitchfork and transcritical curves, (c) (or bottom figure) bifurcation diagram  as the parameter $c$ is varied.   Here $\gamma=0.65450, \alpha=0.1, \beta=0.3, \delta=0.26, c=0.09917, \xi=1.97172$. (Note: P=Pitchfork point, NS=Neutral Saddle point, PTC=Pitchfork-transcritical point.)}
\label{fig:Pitchfork-Transcritical}
\end{figure}

\begin{figure}[!htb]
\begin{center}
    \includegraphics[scale=.44]{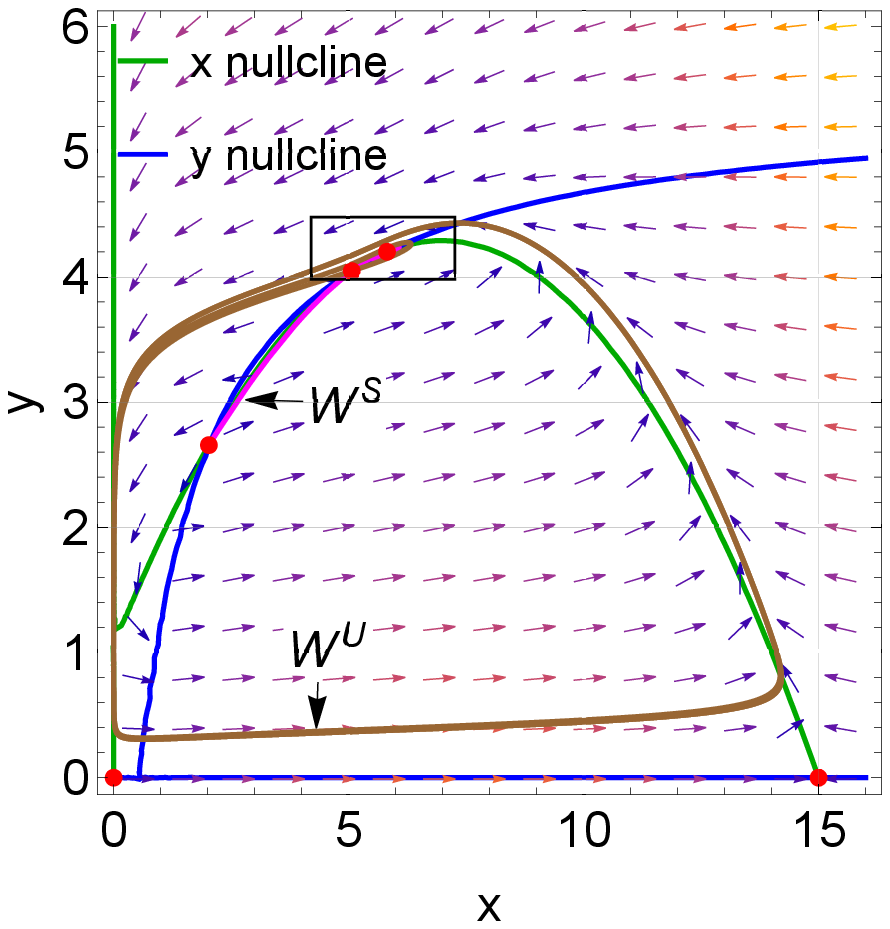}
    \includegraphics[scale=.47]{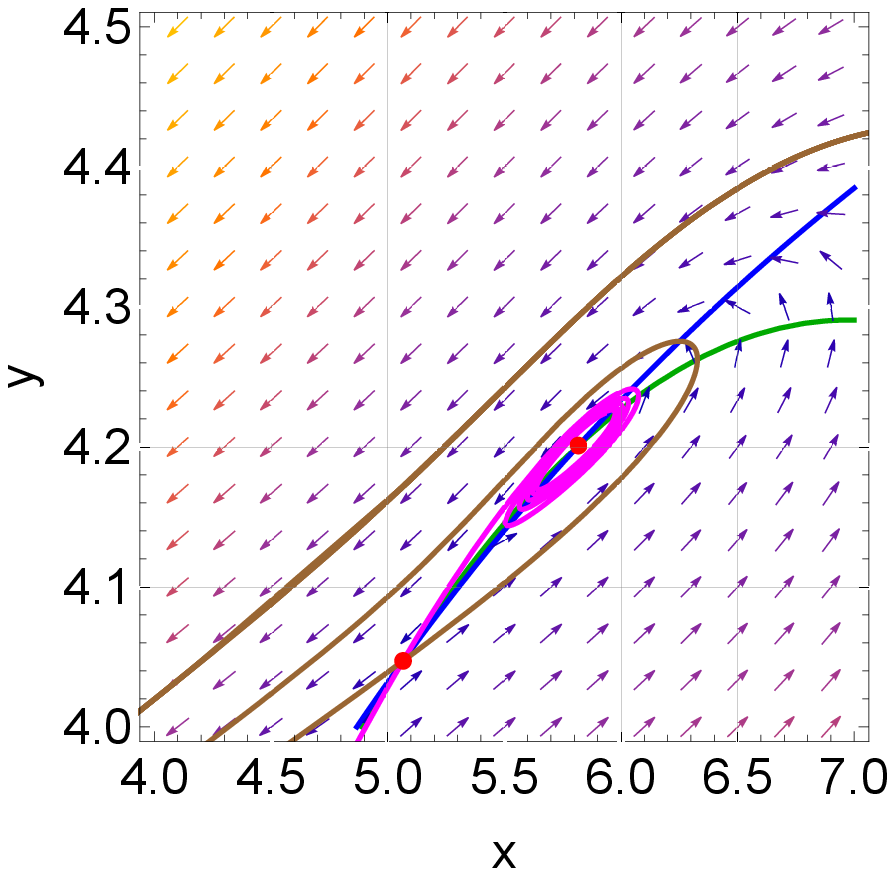}
    \includegraphics[scale=.47]{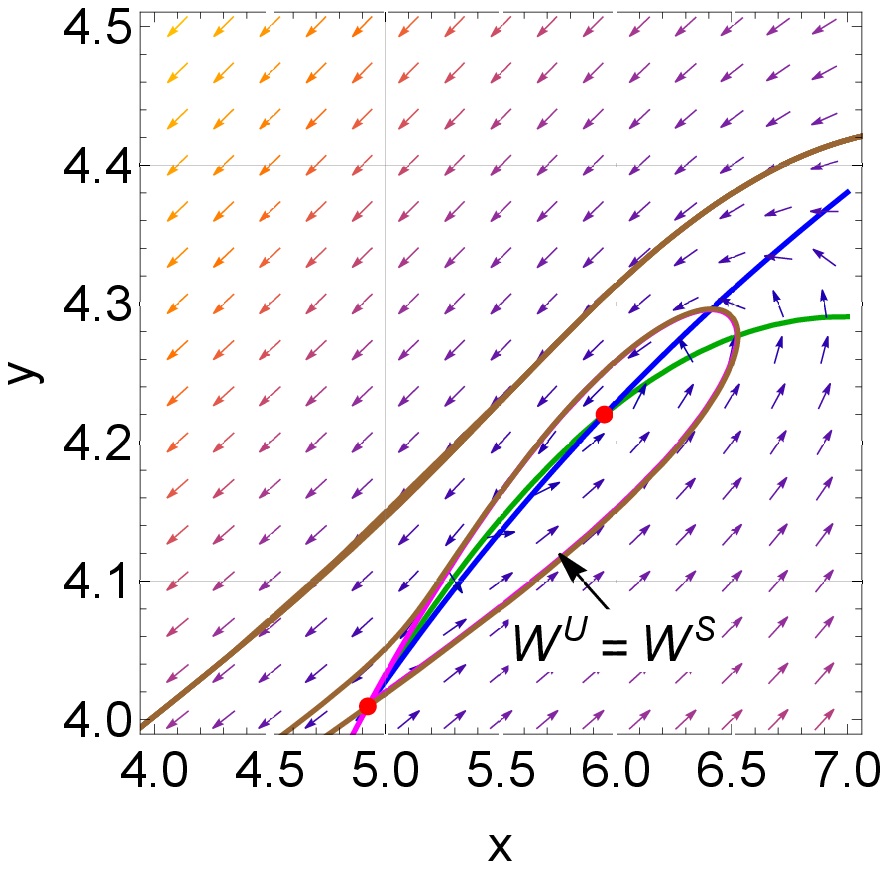}    
\end{center}
\caption{(a) (or left figure) Multiple limit cycles observed for $c=0.069287$. (b) (or middle figure) Zoomed in rectangle in (a) showing the small unstable limit cycle containing the stable interior equilibrium point. (c) (or right figure) Homoclinic loop is formed when $W^S=W^U$ for $c=0.069358$. Here $\gamma=15, \xi=0.45, \alpha=0.1, \beta=0.45, \delta=0.28$.}
\label{fig:limit cycles and homoclinic}
\end{figure}

\section{Turing Instability}

We consider the spatially explicit model for $u \in \Omega$, $\Omega \subset \mathbb{R}^{n}, n=1,2$, with $|\Omega| < \infty$. 

\begin{equation}
\label{Eqn:1nnpd}
\frac{\partial x}{\partial t} = d_{1}\Delta x + x\left(1-\frac{x}{\gamma}\right) - \frac{xy}{1+\alpha \xi + x}, \   
 \frac{\partial y}{\partial t} =d_{2}\Delta y + \frac{\beta xy}{1+\alpha \xi + x} + \frac{\beta \xi y}{1+\alpha \xi + x} - \delta y - c\xi y^{2},
\end{equation}

\begin{equation}
\label{Eqn:1nnpd1b}
\nabla x \cdot \eta = \nabla y \cdot \eta = 0,
\end{equation}

\begin{equation}
\label{Eqn:1nnpd1i}
x(u,0)=x_{0}(u), y(u,0)=y_{0}(u),
\end{equation}
where $d_1$ and $d_2$ are diffusion coefficients.
Note first,
\begin{lemma}
\label{lem:l1233}
Consider the reaction diffusion system \eqref{Eqn:1nnpd}-\eqref{Eqn:1nnpd1i}, when $c=0$, that is when there is no additional food mediated competition amongst the predators. Then there is no Turing instability in the system for any range of parameters or diffusion coefficients.
\end{lemma}

\begin{proof}
A standard linearization of \eqref{Eqn:1nnpd}-\eqref{Eqn:1nnpd1i}, when $c=0$ yields,

\begin{center}
$J=
\begin{bmatrix}
1-\frac{2x^*}{\gamma}-\frac{\left(1+\alpha\xi \right)y^*}{\left(1+\alpha\xi+x^*\right)^2} &  -\frac{x^*}{\left(1+\alpha\xi+x^*\right)} \\
\frac{\left(1+\alpha\xi-\xi\right)\beta y^*}{\left(1+\alpha\xi+x^*\right)^2} & \frac{\beta\left(x^*+\xi\right) }{\left(1+\alpha\xi+x^*\right)}-\delta
\end{bmatrix}
$.\\
\end{center}

\noindent Clearly $J_{22} = 0$, hence $J_{11}J_{22} = 0$ and so the necessary condition required for Turing instability  (that is $J_{11}J_{22} < 0$), is violated. This proves the Lemma.

\end{proof}

We first state the following result.

\noindent Suppose ($x^*,y^*$) is a stable interior equilibrium point.
We shall investigate the possibility of destabilizing this stable interior equilibrium point via  unequal diffusions. This will demonstrate the existence of Turing instability caused by additional food mediated competition, which is clearly not possible without competition via Lemma \ref{lem:l1233}.

The system \eqref{Eqn:1nnpd} can be rewritten as 
\begin{equation}
\label{Eqn:pdelinear}
\frac{\partial X}{\partial t}=D\Delta X + JX,
\end{equation}
where $X=[x~~ y]$ and D=diag($d_1,d_2$) and the Jacobian $J$ is defined in \eqref{Eqn:JacobianMain}.

We analyze the system with Neumann conditions on the boundaries on a spatial domain $\Omega=[0,L] \times [0,L]$ and the wave number $k = n\pi/L$.
Since ($x^*,y^*$) is a stable interior equilibrium point, when there is no diffusion present i.e. $D=0$, we have the following condition:
$\Tr(J) = J_{11}+J_{22}<0$ and $\det(J)= J_{11} J_{22}-J_{12} J_{21}>0$. \\
The solution to the system \eqref{Eqn:pdelinear} is of the form 
\begin{equation}
\label{Eqn:pdesolution}
X(\Omega,t)=\sum_{k}c_k \exp^{\lambda (k) t}X_k(\Omega),
\end{equation}
where $c_k$ are amplitudes determined by Fourier expansion and $X_k(\Omega)$ are the eigenvectors dependent on the eigenvalues $\lambda(k)$ of the eigenvalue problem where $k=n\pi/L$. So to obtain Turing instability we need atleast one positive eigenvalue of the system \eqref{Eqn:1nnpd} 
\begin{equation}
\label{Eqn:eigeneq}
[\lambda \bold{I} - \bold{J} - k^2\bold{D}]X_k=0, 
\end{equation}
where $\lambda$ denotes the eigenvalue. 
For simplicity we normalize the diffusion coefficients to obtain D=diag(1,d). From condition  \eqref{Eqn:eigeneq}, we obtain that the eigenvalues are the roots of the equation
\begin{equation}
\label{Eqn:eigeneq2}
\lambda^2+\lambda[k^2(1+d)-Trace(\bold{J})]+h(k^2)=0, 
\end{equation}
where
\begin{equation}
\label{Eqn:h(keq)}
h(k^2)=dk^4-(dJ_{11}+J_{22})k^2+det(\bold{J}).
\end{equation}

Now to get a positive eigenvalue which shows instability due to diffusion, the following must hold 
\begin{equation}
\label{Turcond1}
dJ_{11}+J_{22}>0, 
\end{equation}
from \eqref{Turcond1} and the fact that $\Tr(J)<0$, we can conclude that we need $J_{11}J_{22}<0$.
\begin{equation}
\label{Turcond2}
(dJ_{11}+J_{22})^2-4d det(\bold{J})>0,
\end{equation}
and for a fixed value of $d$ the value of $k$ can be determined by $h(k^2)<0$. Thus we state the following lemma,

\begin{lemma}
\label{lem:l12}
Consider the reaction diffusion system \eqref{Eqn:1nnpd}-\eqref{Eqn:1nnpd1i}, when $c>0$, that is when there is additional food mediated competition amongst the predators. Then there is Turing instability in the system for some range of parameters or diffusion coefficients.
\end{lemma}

We proceed to perform numerical simulations next.

\subsection{Numerical simulation}

We analyze the numerical simulation to demonstrate the Turing instability. For the simulations we have used PDEPE function in MATLAB. We have performed a time series analysis over a long time period (i.e. $t=5000$) to obtain the values of the stable equilibrium point ($x^*,y^*$) from the set of carefully chosen parameters. The parameter set chosen to have a stable equilibrium point is $\gamma=11, \xi=.3 , \beta=3.9, \delta=1, c=0.9,$ and $\alpha=1$. Clearly $\xi<\frac{\delta}{\beta- \delta \alpha}$ and we get $P(x^*,y^*)\approx(0.2684,1.5297)$ at $t=5000$.
   $P$ is a stable interior equilibrium point with $J_{11}\approx0.1428;J_{12}\approx-0.2484;J_{21}\approx2.4254;J_{22}\approx-0.4127$. Thus $\Tr(J)\approx-0.27<0$ and $\det(J)\approx.5435>0$. Furthermore $J_{11}J_{22}<0$ since $J_{11}>0$ and $J_{22}<0$. All the simulations are performed on spatial grid with $200,500,1000$ and $2000$ data points on the domain $\Omega$ where $L=30\pi$. For some carefully chosen diffusion coefficients $d_1=1$ and $d_2=100$ for which conditions \eqref{Turcond1} and \eqref{Turcond2} are satisfied, there exists at least one node $k>0$ for which we get Turing instability.
     We provide numerical simulation to corroborate Turing instability in Fig \ref{fig:Turing pattern}. The initial condition used is a perturbation about the steady state $P$ by a function of the form $0.0001\cos(0.426x)$ where the value of $k$ was chosen for $h(k^2)<0$ in Fig \ref{fig:Turing pattern}A. We have ran simulations for different perturbations and noticed identical pattern as depicted in Fig \ref{fig:Turing pattern}C. The simulations are run for different values of time, in particular, $t =1000,2000,3000,5000$ and $10,000$. The same time series pattern is seen in Fig \ref{fig:Turing pattern}B.
    
    \begin{figure}
\begin{center}
    \includegraphics[scale=.21]{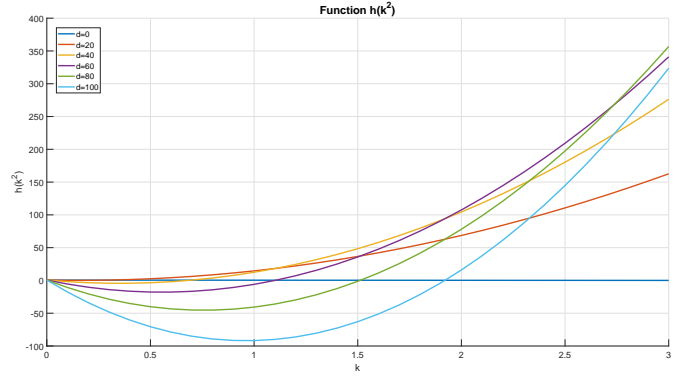}
    \label{fig:tur1}
    \subcaption{}

    \includegraphics[scale=.21]{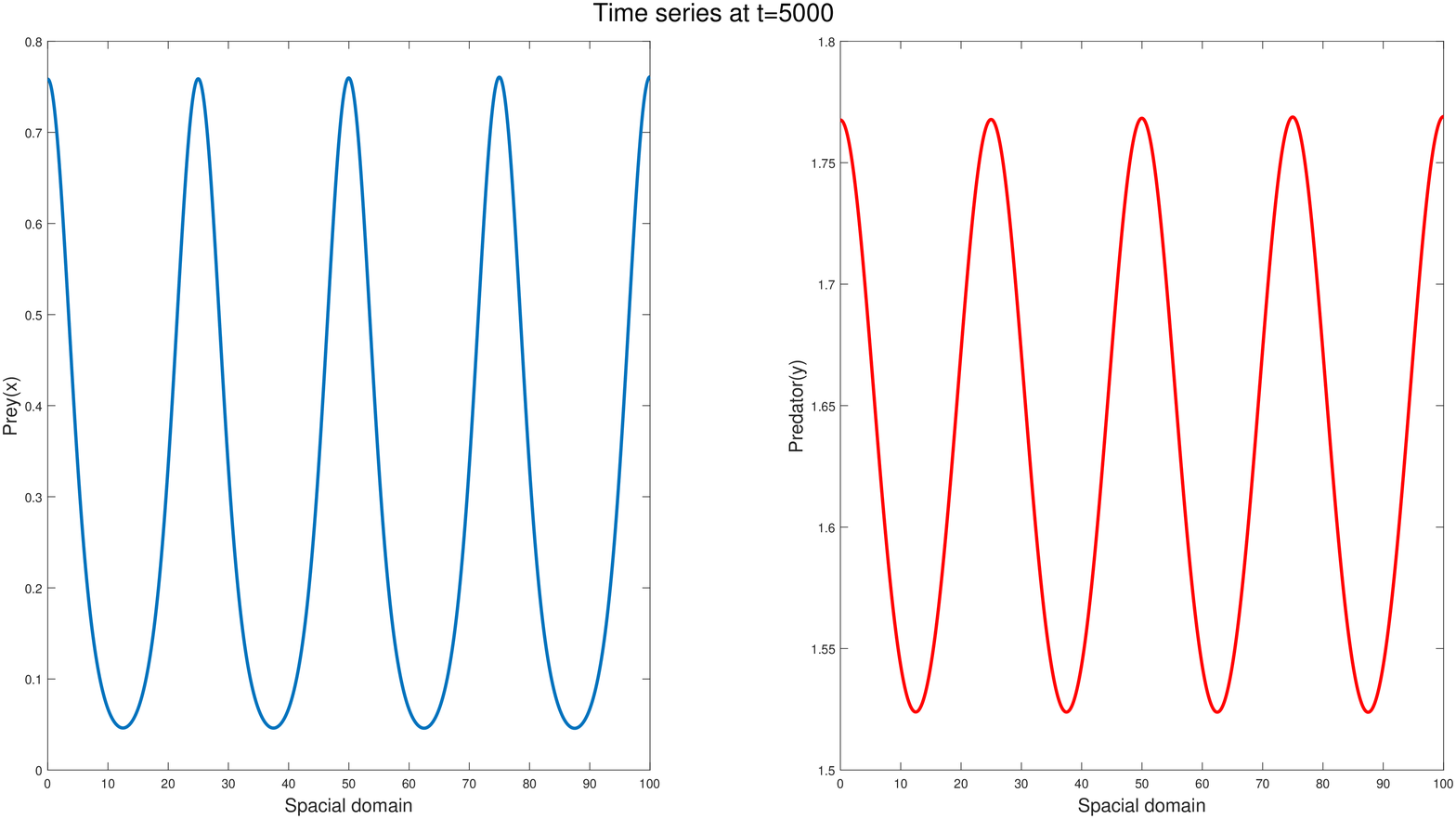}
    \label{tur2}
\subcaption{}
    \includegraphics[scale=.21]{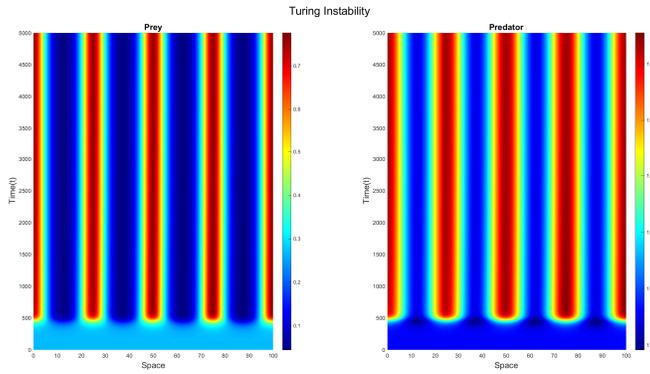}
    \label{tur3}
 \subcaption{} 
\end{center}
\caption{Here we demonstrate that additional food competition can induce Turing stability. We see Turing pattern in both prey and predator. In Fig A we see the functional curve $h(k^2)$ for different diffusion constant which helps in determining the wave number k when $h(k^2)<0$ . In Fig B we see the progression of prey and predator with time till t=5000. In Fig C we see the spatial Turing pattern in prey and predator over time and spatial domain which shows the destabilization of a equilibrium point due to diffusion. }
\label{fig:Turing pattern}
\end{figure}

\section{Discussion and Conclusions}
\label{conj}

We have shown that increasing additional food quantity $\xi$ in predator-pest models, with Holling type II functional response, s.t. $\xi > \xi_{critical}$, can cause unbounded growth of the introduced predator, and so can \emph{hinder} bio-control, by causing a host of non-target effects, associated with excessively high predator density. We Conjecture,

\begin{conjecture}
\label{conj:c11}
Consider the general predator-pest model \eqref{Eqn:1g}. If the quantity of additional food $\xi$ is s.t., $\xi > \xi_{critical}$, then the predator population $y$, blows up in infinite time, for any pest dependent functional response $f$.
\end{conjecture}

One could also model the ``competition" term via  $c \xi y^{p}, 1 \leq p \leq 2$, and then it could have different ecological connotations. If $p=1$, one obtains classical linear harvesting, and if $p=2$, one has classical natural intraspecific competition. If $1<p<2$, one has a nonlinear form of harvesting (or possibly density dependent death). Irrespective, of the form of modeling competition among predators, we clearly see that this increases the possibilities of forms of bio-control. In particular, the pest free state is now stable (or at least locally stable) under a much richer class of dynamics. We observe and report several non-standard bifurcation phenomenon, see Fig. \ref{fig:Cusp-Transcritical} - Fig. \ref{fig:Turing pattern}.
Albeit, some of these are constructed due to the special symmetry we have around the pest free state, where one of the equilibrium points is negative. Thus this begets the question if we could obtain a PTC or CPTC in the interior. This would make for interesting future investigations.

Another interesting direction would be to derive our predator-pest system, using the classical optimal foragaing theory in the setting of one predator - two prey \cite{K96, K10}, where the second prey item would be the additional food. However, in this framework, the handling time of the preys do not depend on each other. In our assumed frame work, we assume the additional food \emph{influences} the handling time of target prey, and working to verify this possibility both theoretically and via laboratory experiments, are the subject of current and future investigations. Note that if we allow a very large amount of additional food to be input into the system, it is unreasonable to expect that the predator will continue to focus on the target pest. Thus a reasonable assumption here may be to restrict the range of $\xi$, so as to be in tune with the classical one predator - two prey optimal foraging theory.  A key to pest eradication seems clear - increase per capita prey consumption rates when the pest population size is low. In this regard, exploration of other functional forms for the density-dependent food supplementation may prove fruitful, particularly ones that describe changes in predator handling and searching behavior that depend on pest and/ or additional food densities.

Mathematical models of biocontrol can provide an idea of what dynamics are possible and suggest routes by which pest eradication is theoretically feasible. Future directions involve studying the effects of pest refuge, evolutionary effects, fear effects, further competitive effects, as well as stochastic effects \cite{PQB16, K10, PK05, B13, B15, F14, PB16, H93, PAT21, BOP22, TM22}. However, experimental tests are required to assess the biological reality of applying these strategies. Laboratory experiments using dynamically interacting predator and pest are encouraged, and would provide an additional intermediate step linking theory to successful bio-control applications in the natural world.

\end{document}